\RequirePackage[l2tabu, orthodox]{nag}		
\documentclass[10pt]{article}

\usepackage[T1]{fontenc}
\usepackage[frenchmath]{mathastext}							
\usepackage{natbib}       									
\usepackage{amsmath}                        
\numberwithin{equation}{section}
\usepackage{graphicx} 											
\usepackage{adjustbox,float}
\usepackage{enumitem}												
\usepackage{mdwlist}												
\usepackage[dvipsnames]{xcolor}							
\usepackage[plainpages=false, pdfpagelabels]{hyperref} 
\hypersetup{
colorlinks   = true,
citecolor    = RoyalBlue, 
linkcolor    = RubineRed, 
urlcolor     = Turquoise
}

\usepackage[displaymath, mathlines]{lineno}
\usepackage{amssymb}                        
\usepackage[mathscr]{eucal}                 
\usepackage{dsfont}													
\usepackage[paperwidth=8.5in,paperheight=11in,top=1.25in, bottom=1.25in, left=1.0in, right=1.0in]{geometry} 
\usepackage{mathtools}                      
\mathtoolsset{showonlyrefs=true}          
\usepackage{amsthm}                         
\allowdisplaybreaks                       
\theoremstyle{plain}
\newtheorem{theorem}{Theorem}
\numberwithin{theorem}{section}
\newtheorem{lemma}{Lemma}       	
\numberwithin{lemma}{section}
\newtheorem{proposition}{Proposition}
\numberwithin{proposition}{section}
\newtheorem{corollary}{Corollary}
\numberwithin{corollary}{section}
\theoremstyle{definition}
\newtheorem{definition}{Definition}
\numberwithin{definition}{section}

\numberwithin{example}{section}
\newtheorem{remark}{Remark}
\numberwithin{remark}{section}

\numberwithin{assumption}{section}

\newcommand\Eb{\mathds{E}}

\newcommand\Ib{\mathds{1}}

\newcommand\Pb{\mathds{P}}

\newcommand\Rb{\mathds{R}}

\newcommand\Ac{\mathscr{A}}
\newcommand\Bc{\mathscr{B}}
\newcommand\Cc{\mathcal{C}}
\newcommand\Dc{\mathscr{D}}

\newcommand\Fc{\mathscr{F}}

\newcommand\Lc{\mathscr{L}}

\newcommand\eps{\varepsilon}

\newcommand\Om{\Omega}
\newcommand\sig{\sigma}

\newcommand\gam{\gamma}

\newcommand\lam{\lambda}
\newcommand\del{\delta}

\newcommand\kap{\kappa}

\newcommand\xu{\underline{x}}

\newcommand\fv{\mathbf{f}}

\newcommand\yv{\mathbf{y}}

\newcommand\uv{\mathbf{u}}
\newcommand\wv{\mathbf{w}}

\newcommand\ct{\tilde{c}}

\newcommand\yt{\widetilde{y}}

\newcommand\etat{\widetilde{\eta}}
\newcommand\deltat{\tilde{\del}}

\newcommand\pht{\widetilde{\varphi}}

\newcommand\dd{\mathrm{d}}
\newcommand\ee{\mathrm{e}}

\newcommand{\cs}{{c^*}}

\newcommand\pis{{\pi^*}}
\newcommand\thts{{\theta^*}}

\newcommand{\xa}{{x^{*}}}

\newcommand{\ya}{y^*}

\newcommand{\R}{\mathbb{R}}

\newcommand{\esup}{\operatorname{ess}\sup}

\newcommand{\F}{\mathbb{F}}
\newcommand \al {\alpha}
\newcommand \bet {\beta}
\newcommand \vp {v}
\newcommand \vpt {v}

\providecommand{\keywords}[1]{\noindent\textbf{\textit{Keywords: }} #1}

\begin{document}

\title{Optimal Investment and Consumption under a Habit-Formation Constraint}

\author{
Bahman Angoshtari
\thanks{Department of Mathematics, University of Miami.  \textbf{e-mail}: \url{bangoshtari@miami.edu}}
\and
Erhan Bayraktar
\thanks{Department of Mathematics, University of Michigan.  \textbf{e-mail}: \url{erhan@umich.edu} E. Bayraktar is supported in part by the National Science Foundation under grant DMS-1613170 and by the Susan M. Smith Professorship.}
\and
Virginia R. Young
\thanks{Department of Mathematics, University of Michigan.  \textbf{e-mail}: \url{vryoung@umich.edu} V. R. Young is supported in part by the Cecil J. and Ethel M. Nesbitt Professorship.}
}

\date{This version: \today}

\maketitle

\begin{abstract}
We formulate an infinite-horizon optimal investment and consumption problem, in which an individual forms a habit based on the exponentially weighted average of her past consumption rate, and in which she invests in a Black-Scholes market.
The individual is constrained to consume at a rate higher than a certain proportion $\alpha$ of her consumption habit. Our habit-formation model allows for both addictive ($\alpha=1$) and nonaddictive ($0<\alpha<1$) habits. The optimal investment and consumption policies are derived explicitly in terms of the solution of a system of differential equations with free boundaries, which is analyzed in detail. If the wealth-to-habit ratio is below (resp.\ above) a critical level $x^*$, the individual consumes at (resp.\ above) the minimum rate and invests more (resp.\ less) aggressively in the risky asset. Numerical results show that the addictive habit formation requires significantly more wealth to support the same consumption rate compared to a moderately nonaddictive habit. Furthermore, an individual with a more addictive habit invests less in the risky asset compared to an individual with a less addictive habit but with the same wealth-to-habit ratio and risk aversion, which provides an explanation for the equity-premium puzzle.   
\end{abstract}


%
%

\section{Introduction}\label{sec:intro}

The study of consumption habit formation is a classical topic in financial economics and the literature goes back to the late 1960's. See, for instance, \cite{Pollak1970}, \cite{RyderHeal1973}, \cite{sundaresan1989}, \cite{Constantinides1990}, \cite{DetempleZapatero1991}, \cite{DetempleZapatero1992} for early works, and \cite{DetempleKaratzas2003}, \cite{MUNK2008}, \cite{EnglezosKaratzas2009}, \cite{Muraviev2011}, and \cite{2015Yu} for more recent studies. In this literature, habit formation is modeled through the so-called \emph{habit-formation preference} $\Eb\left[\int_0^T U(t, C_t-Z_t)\dd t\right]$, in which $U:[0,T]\times\Rb\to \Rb$ is a given utility function and $Z_t$ is the agent's \emph{habit} (or standard of living) defined as the exponentially weighted running average of past consumption rates $C_s$, $0\le s< t$. If the consumption rate is allowed to fall below the habit, the habit-formation model is called \emph{nonaddictive}. Otherwise, a model with a constraint $C_t\ge Z_t$ is called \emph{addictive}. habit-formation models are notoriously more difficult to solve than their non-habit formation counterparts. Indeed, explicit forms for optimal policies are rare and, in most cases, the optimal policy is specified in terms of a solution of a PDE or an unknown process characterized via the martingale representation theorem.

A related literature on consumption ratcheting and drawdown is devoted to models of optimal consumption under a more severe form of habit formation in which the reference point for forming habit is the running maximum of past consumption rates (instead of their running average). \cite{Dybvig1995} found the optimal investment and consumption policies for an investor in a Black-Scholes financial market who seeks to maximize discounted utility of consumption, while imposing a ratcheting constraint on the rate of consumption (that is, the consumption rate has to be a non-decreasing process). \cite{Arun2012} extended \cite{Dybvig1995} by allowing the rate of consumption to decrease, but not below a fraction of its maximum rate (that is, a so-called drawdown constraint on the consumption rate). See, also, \cite{jeon2018portfolio} and \cite{2019Roche} for similar models. \cite{AngoshtariBayraktarYoung2019} solved a problem setting similar to that of \cite{Arun2012} that also allowed for agent's bankruptcy (in the context of an optimal dividend problem), which occurred with positive probability. In a related yet different setting, \cite{AlbrecherAzcueMuler2020} and \cite{AlbrecherAzcueMuler2021} considered an optimal dividend problem in a Brownian risk model while imposing a ratcheting constraint on the dividend rates. In the studies above, habit formation is modeled by imposing a constraint on admissible consumption policies, rather than through the objective function, which is the approach taken for classical habit-formation models. Recently, \cite{DengLiPhamYu2020} provided a direct link to the classical literature of habit formation by solving an optimal investment and consumption model with a habit-formation preference (that is, they modeled habit formation through the objective function rather than through the admissibility set), in which the habit is presented by the running maximum of consumption.

Habit-formation models based on the running maximum have been more tractable and produced more explicit policies than those with the running average as the reference point. The former class of models, however, represent a more extreme form of habit formation in the sense that the effect of past consumption does not ``fade away'' with time, as one expects.  Indeed, under a drawdown constraint, our future habits will change forever if we decide to increase consumption beyond its historical maximum.  In reality, recent levels of consumption have more effect on our current consumption habit than how we consumed a long time ago, and the effect of past consumption fades away with time. These observations motivated us to consider a habit-formation model in which the reference point of habit is the running average of consumption (as in the habit-formation literature), and the habit-formation mechanism operates though a constraint on admissible consumption policies (as in the consumption ratcheting and drawdown literature). In a sense, we also provide a connection between these two bodies of work, however in the opposite direction of \cite{DengLiPhamYu2020}.  

In \cite{AngoshtariBayraktarYoung2020}, we provided the first step by solving a deterministic optimal consumption problem with the objective of maximizing the functional $\int_0^{+\infty} \ee^{-\del t} \frac{\big[C(t)/Z(t)\big]^{1-\gam}}{1-\gam}\dd t$ while imposing the habit-formation constraint $C(t)\ge \al Z(t)$ for all $t\ge0$. Here, $C(t)$, $t\ge0$, is the deterministic consumption rate and
\begin{linenomath*}\begin{align}
Z(t) = \ee^{-\rho\,t}\left(z + \int_0^{t} \rho\,\ee^{\rho\,u} C(u) \dd u\right);\quad t\ge0,
\end{align}\end{linenomath*}
is the agent's habit at time $t$. In particular, we assumed that the individual funds her consumption solely through a riskless asset offering an interest rate $r>0$; thus, wealth and consumption processes were deterministic. To avoid bankruptcy, we showed that the wealth-to-habit ratio must always be above a certain level $\xu$ given by \eqref{eq:xu} below. We showed that there exists a threshold $\xa$ such that if the ratio of wealth-to-habit is above (resp.\ below) $\xa$, it is optimal to consume at a rate greater than (resp.\ equal to) the minimum acceptable rate imposed by the habit-formation constraint. We also found a significant difference between impatient individuals (those with $\del\ge \rho(1-\al)+r$) and patient individuals (those with $0<\del<\rho(1-\al)+r$). Impatient individuals always consume above the minimum rate (that is, $\xa = \xu$) and, thereby, eventually attain the minimum wealth-to-habit ratio $\xu$, while patient individuals might consume at the minimum rate (that is, $\xa > \xu$) and, thereby, attain a wealth-to-habit ratio greater than the minimum acceptable level.  We obtained explicit results in terms of the solution of a nonlinear free-boundary problem.

In this paper, we extend the model in \cite{AngoshtariBayraktarYoung2020} by assuming that the agent invests in a Black-Scholes financial market. We formulate and solve a stochastic control problem to obtain the optimal investment and consumption policies. We find that the optimal consumption policy has a similar general structure as what we found in the riskless case. That is, there exists a critical level $\xa$ of wealth-to-habit ratio such that the agent consumes above the minimum rate if her wealth-to-habit ratio is above $\xa$ and consumes at the minimum rate otherwise. The value of $\xa$ and the optimal consumption function are, however, different from their counterparts in the riskless case. In particular, we don't see the structural difference between the consumption functions of patient and impatient individuals in that $\xa>\xu$ for all values of $\del$. As for the investment policy, we found that the agent optimally invests ``more aggressively'' in the stock when her wealth-to-habit ratio is below $\xa$ compared to when it is above $\xa$. By more aggressive investment, we mean that an (infinitesimal) increase in wealth-to-habit results in a larger increase in stock's holdings. Finally, numerical analysis shows that increasing $\al$ (while keeping wealth-to-habit ratio and risk aversion constant) decreases the optimal investment in the risky asset. In other words, individuals with more addictive habit formation (that is, larger $\al$) optimally invest \emph{less} in the risky asset. Thus, the market has to provide a higher premium to attract such an individual which indicates that our model provides an explanation for the equity premium puzzle of \cite{mehra1985equity}.

On the mathematical side, the results presented here rely on analyzing a coupled system of first-order ODEs with a free boundary, as opposed to a single ODE in \cite{AngoshtariBayraktarYoung2020}. The analysis of such a system is more delicate (see Proposition \ref{prop:FBP-sol}) and provides the main technical backbone of the paper. A second technical point of the paper is the verification theorem (Theorem \ref{thm:VF}), which did not pose many difficulties in \cite{AngoshtariBayraktarYoung2020} when there is no stochasticity involved. Besides, the fact that the drift coefficient of the optimal wealth SDE has more than linear growth and the coefficients are only semi-explicit makes certain parts of the verification argument somewhat non-standard.

The paper is organized as follows. In Section \ref{sec:setup}, we introduce the consumption habit process and its basic properties, formulate a stochastic control problem for finding the optimal investment and consumption policy, and prove a verification lemma for the stochastic control problem. In Section \ref{sec:optimal}, we formulate the Hamilton-Jacobi-Bellman (HJB) free-boundary-problem and solve it semi-explicitly by applying the Legendre transform. This section also includes the main result of the paper, namely, Theorem \ref{thm:VF}, in which we verify that the solution of the HJB free-boundary problem yields the value function and the optimal investment and consumption policies. In Section \ref{sec:numerics}, we include a series of numerical examples that highlight certain properties of the optimal policy. Proofs of auxiliary results are included in Appendices \ref{app:verif} and \ref{app:aux_constrained}.
%
%

\section{Problem formulation}\label{sec:setup}

We consider an individual who invests in a market consisting of a riskless and a risky asset in order to maximize her utility of lifetime consumption. We assume that the riskless asset pays interest at a fixed rate $r>0$ and that the price of the risky asset $(S_t)_{t\ge0}$ follows a geometric Brownian motion
\begin{linenomath*}\begin{align}
\frac{\dd S_t}{S_t} = \mu \dd t + \sig \dd B_t;\quad t \ge 0.
\end{align}\end{linenomath*}
Here, $\mu > r$ and $\sig > 0$ are constants, and $(B_t)_{t\ge0}$ is a standard Brownian motion in a filtered probability space $\big(\Om, \Fc, \Pb, \F = (\Fc_t)_{t\ge0}\big)$, in which the filtration $\F$ is generated by the Brownian motion and satisfies the usual conditions.

Let $\pi_t$ denote the amount invested in the risky asset, and let $C_t$ denote the individual's consumption rate at time $t$, so that $\int_0^t C_u\dd u$ is the total consumption over the time interval $[0,t]$. Then, her wealth process $(W_t)_{t\ge0}$ follows the dynamics
\begin{linenomath*}\begin{align}\label{eq:wealth}
\dd W_t &= \big(r\,W_t + (\mu-r) \pi_t - C_t\big)\dd t + \sig \pi_t \dd B_t,
\end{align}\end{linenomath*}
for $t\ge0$, with $W_0=w>0$.

For a given consumption process $(C_t)_{t\ge0}$, we define the individual's \emph{habit process} (that is, consumption habit) as the process $(Z_t)_{t\ge0}$ given by
\begin{linenomath*}\begin{align}\label{eq:Z}
Z_t = \ee^{-\rho\,t}\left(z + \int_0^{t} \rho\,\ee^{\rho\,u} C_u\dd u\right);\quad t\ge0,
\end{align}\end{linenomath*}
which has the following equivalent differential form:
\begin{linenomath*}\begin{align}\label{eq:Z2}
\begin{cases}
	\dd Z_t = -\rho(Z_t - C_{t})\dd t;\quad t\ge0,\\
	Z_0=z.
\end{cases}
\end{align}\end{linenomath*}
Here, $\rho>0$ is a constant, and $z>0$ represents the initial consumption habit of the individual.  The parameter $\rho$ determines how much current habit is influenced by the recent rate of consumption relative to the consumption rate farther in the past.  As $\rho$ increases, more weight is given to recent consumption.  In the limiting cases, $\rho = 0$ implies $Z_t = z$, and $\rho = + \infty$ implies $Z_t = C_{t}$.

For $t>0$, the consumption habit $Z_t$ given by \eqref{eq:Z} is the exponentially weighted moving average of past consumption $(C_s)_{s<t}$. To see this, let us assume that the individual lived (and consumed) over the time period $(-\infty, t)$. Let $z$ be the exponentially weighted average of her consumption rate before time zero, that is, $z = \int_{-\infty}^0 \rho\,\ee^{\rho u} C_u \dd u$.
(Note that $\int_{-\infty}^0 \rho \, \ee^{\rho u} \dd u = 1$.)  By substituting for $z$ in \eqref{eq:Z}, we obtain
\begin{linenomath*}\begin{align}
Z_t &= \int_{-\infty}^0 \rho\,\ee^{-\rho(t-u)} C_u\dd u + \int_0^{t} \rho\,\ee^{-\rho(t-u)} C_u\dd u\\
&= \int_{-\infty}^{t} \rho\,\ee^{-\rho(t-u)} C_u\dd u,
\end{align}\end{linenomath*}
with $\int_{-\infty}^{t} \rho\,\ee^{-\rho(t-u)} \dd u = 1$.
Thus, $Z_t$ is the exponentially weighted moving average of $(C_s)_{s<t}$, as claimed.

We consider a \emph{consumption habit formation} for the individual by assuming that, at any time $t\ge0$, she is unwilling to consume at a rate that is below a certain proportion of her  habit $Z_t$. In particular, we impose the following constraint on the individual's consumption process
\begin{linenomath*}\begin{align}\label{eq:Habit}
C_t \ge \al\,Z_t;\quad \Pb{\text -}a.s., \, t \ge 0,
\end{align}\end{linenomath*}
in which $0< \al\le 1$ is a constant that measures the individual's tolerance for her current consumption to drop below her habit.  The larger the value of $\al$, the less tolerant the individual is in allowing her current consumption to fall below her habit. Note that the consumption habit process $(Z_t)_{t\ge0}$ depends on $z$ and on the consumption process $(C_t)_{t\ge0}$. To ease the notational burden, however, we write $Z_t$ instead of the more accurate $Z_t^{z, (C_s)_{0 \le s \le t}}$. 

The following lemma establishes a lower bound for the consumption habit process, and we use it in later arguments.  We omit the proof of this lemma because it closely follows the proof of Lemma 2.1 in \cite{AngoshtariBayraktarYoung2020}.

\begin{lemma}\label{lem:Z-Lbound}
Let $(C_t)_{t\ge0}$ be a consumption process satisfying \eqref{eq:Habit}, in which $(Z_t)_{t\ge0}$ is given by \eqref{eq:Z}. We, then, have
\begin{linenomath*}\begin{align}\label{eq:Z-Lbound}
	Z_t\ge Z_s \ee^{-\rho(1-\al)(t-s)},
\end{align}\end{linenomath*}
$\Pb$-a.s., for all $0\le s \le t$. In particular, $Z_t \ge z \ee^{-\rho(1-\al)t}$, $\Pb$-a.s., for all $t \ge 0$.   \qed
\end{lemma}

We assume that the individual avoids bankruptcy with probability one. The following lemma provides the corresponding necessary and sufficient condition. In it, we use the notation
\begin{linenomath*}\begin{align}\label{eq:xu}
\xu=\xu(\al):=\frac{\al}{r+\rho(1-\al)}\,,
\end{align}\end{linenomath*}
for $\al\in[0,1]$. Note that $\xu$ is strictly increasing in $\al$, $\xu(0)=0$, and $\xu(1)=1/r$.  Again, we omit the proof of this lemma because it closely follows the proof of Lemma 2.2 in \cite{AngoshtariBayraktarYoung2020}.

\begin{lemma}\label{lem:NoRuin}
Let the $\F$-adapted process $(\pi_t,C_t)_{t\ge0}$  satisfy condition \eqref{eq:integrable0} below, and let $(W_t)_{t\ge0}$ and $(Z_t)_{t\ge0}$ be given by \eqref{eq:wealth} and \eqref{eq:Z}, respectively. Then, $W_t>0$ for all $t\ge0$ $\Pb$-a.s.\ if and only if
\begin{linenomath*}\begin{align}\label{eq:NoRuin}
	\frac{W_t}{Z_t} \ge \xu,
\end{align}\end{linenomath*}
$\Pb$-a.s., for all $t \ge 0$.  \qed
\end{lemma}

Following the proof of Lemma 2.2 in \cite{AngoshtariBayraktarYoung2020}, we provide a detailed discussion of the condition in \eqref{eq:NoRuin}, and we invite the interested reader to refer to that paper.  That said, we repeat that, as $\al \to 0^+$, the requirement for consumption \eqref{eq:Habit} becomes $C_t \ge 0$, and inequality \eqref{eq:NoRuin} becomes moot, which we expect because this limiting case is the market model considered by \cite{Merton1969}.
Also, note that, in the special case of $\al = 1$, the requirement for consumption \eqref{eq:Habit} becomes $C_t \ge Z_t$, and inequality \eqref{eq:NoRuin} becomes $rW_t \ge Z_t$, which is consistent with feasibility condition adapted by \cite{Dybvig1995}, namely, that $rW_t \ge C_{t-}$. Note, however, that our preference specification in \eqref{eq:EU} differs from Merton's and Dybvig's for the case $\al \to 0^+$. Therefore, our optimal policies do not converge to theirs as $\al\to0^+$ or for $\al=1$.

\begin{remark}
	In the classical habit formation preference, the conventional definition of $Z_t$ is $\dd Z_t=-\rho_1Z_t\dd t+\rho_2 C_t\dd t$. When comparing with the classical literature, one should note that the reference point in our model is $\tilde{Z}_t = \al Z_t$ and not $Z_t$. Since $\dd Z_t = -\rho (Z_t - C_t) \dd t$, our reference point satisfies
	\begin{align}
		\dd \tilde{Z}_t = \al \dd Z_t = -\al\rho (Z_t - C_t) \dd t = \big(-\rho \tilde{Z_t} + \al\rho C_t\big)\dd t.
	\end{align}
	Thus, by setting $\rho_1=\rho$ and $\rho_2=\al \rho$, we obtain the same dynamics for the reference point $(\tilde{Z}_t)$ as in the classical literature. More specifically, the case $\al=1$ in our model corresponds to the case $\rho_1=\rho_2$ in the classical literature, while our case $0<\al<1$ corresponds to the case $\rho_1\ne \rho_2$. \qed
\end{remark}

In the following, we define the set of \emph{admissible} investment and consumption policies as those that avoid bankruptcy while satisfying the individual's consumption habit-formation constraint.

\begin{definition}\label{def:Admissible0}
Let $\widetilde{\Ac}(\al)$ be the set of all processes $(\pi_t,C_t)_{t\ge0}$ such that $(\pi_t)_{t\ge0}$ is $\F$-adapted, $(C_t)_{t\ge0}$ is non-negative and $\F$-progressively measurable, 
\begin{linenomath*}\begin{align}\label{eq:integrable0}
	\int_0^t \big(\pi_u^2 + C_u\big)\dd u <+\infty;\quad t\ge0, \, \Pb\text{-a.s.,}
\end{align}\end{linenomath*}
and conditions \eqref{eq:Habit} and \eqref{eq:NoRuin} hold, namely,
\begin{linenomath*}\begin{align}
	C_t \ge \al Z_t,\quad \text{and}\quad W_t\ge \xu Z_t,
\end{align}\end{linenomath*}
$\Pb$-a.s., for all $t \ge 0$, in which $(W_t)_{t\ge0}$ and $(Z_t)_{t\ge0}$ are given by \eqref{eq:wealth} and \eqref{eq:Z}, respectively.  \qed
\end{definition}

Next, we formulate the individual's lifetime consumption and investment problem as a stochastic control problem. For any admissible investment and consumption policy $(\pi_t,C_t)_{t\ge0}$, let us introduce the \emph{wealth-to-habit} process
\begin{linenomath*}\begin{align}\label{eq:X}
X_t := \frac{W_t}{Z_t};\quad t\ge0,
\end{align}\end{linenomath*}
and note that, by \eqref{eq:wealth} and \eqref{eq:Z2},
\begin{linenomath*}\begin{align}\label{eq:X-SDE}
\begin{cases}
	\dd X_t = \Big((\rho+r)X_t + (\mu-r)\theta_t - (1+\rho X_t)c_t\Big)\dd t + \sig\theta_t \dd B_t;\quad t\ge0,\\
	X_0=x:=\frac{w}{z} \ge \xu,
\end{cases}
\end{align}\end{linenomath*}
in which we have defined the \emph{investment-to-habit} process $(\theta_t)_{t\ge0}$ and the \emph{consumption-to-habit} process $(c_t)_{t\ge0}$ by, $\theta_t := \frac{\pi_t}{Z_t}$ and $c_t := \frac{C_t}{Z_t}$,
respectively.

We define the set of \emph{admissible} investment-to-habit and consumption-to-habit policies as follows.
\begin{definition}\label{def:Admissible}
Let $\Ac=\Ac(\al)$ be the set of all processes $(\theta_t,c_t)_{t\ge0}$ such that $(\theta_t)_{t\ge 0}$ is $\F$-adapted, $(c_t)_{t\ge0}$ is $\F$-progressively measurable, 
\begin{linenomath*}\begin{align}\label{eq:integrable}
	\int_0^t \big(\theta_u^2 + c_u\big)\dd u <+\infty;\quad  \Pb\text{-a.s.}, \, t \ge 0,
\end{align}\end{linenomath*}
and
\begin{linenomath*}\begin{align}
	c_t \ge \al,\quad \text{and}\quad X_t\ge \xu,
\end{align}\end{linenomath*}
$\Pb$-a.s., for all $t \ge 0$, in which $(X_t)_{t\ge0}$ is given by \eqref{eq:X}. \qed
\end{definition}

As the following proposition states, our two definitions of admissible policies are equivalent in the sense that any admissible investment and consumption policy corresponds to an admissible relative investment and consumption policy and vice versa. Its proof is an application of It\^o's lemma and, thus, omitted.

\begin{proposition}\label{prop:Admiss}
Assume that $(\pi_t,C_t)_{t\ge0} \in \widetilde{\Ac}(\al)$ and let $(Z_t)_{t\ge0}$ be given by \eqref{eq:Z}. Then, we have  $(\pi_t/Z_t, C_t/Z_t)_{t\ge0}\in\Ac(\al)$.
Conversely, assume that $(\theta_t,c_t)_{t\ge0}\in\Ac(\al)$, and let $(W_t)_{t\ge}$ be the solution of
\begin{linenomath*}\begin{align}
	\begin{cases}\displaystyle
		\frac{\dd W_t}{W_t} = \left(r + (\mu-r) \frac{\theta_t}{X_t} - \frac{c_t}{X_t}\right)\dd t + \sig \, \frac{\theta_t}{X_t} \, \dd B_t;\quad t\ge0,\\
		W_0=w,
	\end{cases}
\end{align}\end{linenomath*}
in which $(X_t)_{t\ge0}$ is given by \eqref{eq:X-SDE}. We, then, have $(\pi_t:=\theta_t W_t / X_t, C_t:=c_t W_t / X_t)\in \widetilde{\Ac}(\al)$.   \qed
\end{proposition}

We assume that the individual values her consumption relative to her habit. In particular, for a given consumption process $(C_t)_{t\ge0}$, the expected utility of her lifetime consumption is given by
\begin{linenomath*}\begin{align}\label{eq:EU}
\Eb\left(\int_0^{\tau_d} \frac{1}{1-\gam}\left(\frac{C_t}{Z_t}\right)^{1-\gam}\, \ee^{-\deltat\,t}\,\dd t\right)
= \Eb\left(\int_0^{+\infty} \frac{1}{1-\gam}\left(\frac{C_t}{Z_t}\right)^{1-\gam}\, \ee^{-(\tilde \lam+\deltat)\,t}\,\dd t\right),
\end{align}\end{linenomath*}
in which $\deltat>0$ is the individual's subjective time preference, $\gam>1$ is her (constant) relative risk aversion, and $\tau_d$ is the random time of her death, which we assume is exponentially distributed with mean $1/\tilde \lam > 0$, and $\tau_d$ is independent of the Brownian motion.

In light of Proposition \ref{prop:Admiss}, the individual's optimal investment-consumption problem is, thus, formulated by the following stochastic control problem:
\begin{linenomath*}\begin{align}\label{eq:VF}
V(x)=V(x,\al) := \sup_{(\theta_t,c_t)\in\Ac(\al)} \Eb_{x}
\left(\int_0^{+\infty} \frac{c_t^{1-\gam}}{1-\gam} \,\ee^{-\del t}\, \dd t\right);\quad x\ge\xu,
\end{align}\end{linenomath*}
in which $\del = \deltat + \tilde \lam$, and $\Eb_{x}$ denotes conditional expectation given $X_0 = x$.

We end this section by proving a verification theorem for the stochastic control problem \eqref{eq:VF}. For its statement,  we define the operator $\Lc_{\theta, c}$ on twice-differentiable functions by
\begin{linenomath*}\begin{align}\label{eq:mL}
\Lc_{\theta, c} v(x) = -\del v(x) +\Big((\rho+r)x + (\mu-r)\theta\Big) v'(x)+\frac{1}{2}\sig^2\theta^2 v''(x) + \frac{c^{1-\gam}}{1-\gam}-c(1+\rho x)v'(x).
\end{align}\end{linenomath*}

\begin{theorem}\label{thm:verif}
Suppose $\vp\in\Cc^{2}\big([\xu,+\infty)\big)$ satisfies the following properties: for any $x \ge \xu$,
\begin{enumerate}
	\item[$(i)$] $\Lc_{\theta,c}\vp(x)\le 0$ for all $\theta\in\Rb$ and $c\ge \al$.
	
	\item[$(ii)$] $v'(x)>0$, $\vp(\xu) = \frac{\al^{1-\gam}}{\del(1-\gam)}$, and $\lim_{x\to\xu^+} v'(x)=+\infty$.
	
	\item[$(iii)$] $\lim_{T\to +\infty} \Eb_{x}\big(\ee^{-\del T} \vp(X_T)\big) = 0$ for any wealth-to-habit process $(X_t)_{t\ge0}$ that arising from an admissible policy $(\theta_t,c_t)_{t\ge0} \in \Ac(\al)$.
	
	\item[$(iv)$] $\Lc_{\thts(x),\cs(x)}\vp(x)=0$ for some functions $\thts(x)$ and $\cs(x)\ge\al$.
	\item[$(v)$] For $\thts$ and $\cs$ in condition $(iii)$, the following stochastic differential equation has a unique strong solution:
	\begin{linenomath*}\begin{align}
		\begin{cases}
			\dd X^*_t = \Big((\rho+r)X^*_t + (\mu-r)\thts(X^*_t) - (1+\rho X^*_t)\cs(X^*_t)\Big)\dd t + \sig\thts(X^*_t) \dd B_t;\quad t\ge0,\\
			X^*_0=x,
		\end{cases}
	\end{align}\end{linenomath*}
	and $\big(\thts(X^*_t), \cs(X^*_t)\big)_{t\ge0}\in\Ac$.
\end{enumerate}
Then, $v = V$ on $[\xu, +\infty)$, and $\big(\thts(X^*_t), \cs(X^*_t)\big)_{t\ge0}$ is an optimal policy.
\end{theorem}

\begin{proof}
See Appendix \ref{app:verif}.
\end{proof}

%
%

\section{Optimal investment and consumption policy}\label{sec:optimal}

In this section, we consider the stochastic control problem \eqref{eq:VF} when \eqref{eq:Habit} is a habit-formation constraint, that is, when $0<\al\le 1$. In other words, we exclude the case $\al=0$. 

Theorem \ref{thm:verif} implies that the value function $V(\cdot \, ;\al)$ is a solution of the following  differential equation:
\begin{linenomath*}\begin{align}\label{eq:HJB}
- \del v(x) + (\rho+r)x v'
+ \sup_{\theta} \left[(\mu-r)\theta v' + \frac{1}{2}\sig^2\theta^2 v''\right]
+\sup \limits_{c\ge \al} \left[ \dfrac{c^{1 - \gam}}{1 - \gam} - (1+\rho x) c v' \right]=0;\quad x\ge\xu.
\end{align}\end{linenomath*}
For the rest of this section, we construct a classical solution of \eqref{eq:HJB}, and then use Theorem \ref{thm:verif} to show that the solution equals the value function $V(\cdot \, ;\al)$ in \eqref{eq:VF}.

To construct a candidate solution, we hypothesize that the optimal investment and consumption policy has the following form. There exists a critical value of wealth-to-habit ratio $\xa\ge\xu$, such that,
\begin{enumerate}
\item[$(a)$] If $\xu\le X_t\le \xa$, it is optimal to consume at the minimum rate, that is, $c_t^* = \al$.  Also, if $X_t=\xu$, it is optimal to invest fully in the riskless asset, that is, $\theta_t^* = 0$.

\item[$(b)$] If $X_t>\xa$, it is optimal to consume more than the minimum rate.
\end{enumerate}
The optimal expressions for $c$ and $\theta$ in \eqref{eq:HJB} are given by
\begin{linenomath*}\begin{align}\label{eq:cs-cand}
\cs(x) :=
\begin{cases}
	\displaystyle
	\al;&\quad (1+\rho x) v'(x) \ge \al^{-\gam},\vspace{1ex}\\
	\displaystyle
	\big((1+\rho x) v'(x) \big)^{-\frac{1}{\gam}};&\quad 0<(1+\rho x)v'(x)< \al^{-\gam},
\end{cases}
\end{align}\end{linenomath*}
and
\begin{linenomath*}\begin{align}\label{eq:thts-cand}
\thts(x) := -\frac{\mu-r}{\sig^2}\frac{\vp'(x)}{\vp''(x)},	
\end{align}\end{linenomath*}
respectively. To obtain these equations, we assume that $\vp_x > 0$ and $\vp_{xx}<0$, which we show in Proposition \ref{prop:HJB} below. Thus, for $(a)$ and $(b)$ in our hypothesis to be true, we must have
\begin{linenomath*}\begin{align}\label{eq:VE}
\begin{cases}
	\displaystyle
	\lim_{x\to\xu^+}\frac{\vp'(x)}{\vp''(x)} = 0,&\vspace{1ex}\\
	(1+\rho x)v'(x) \ge \al^{-\gam};&\quad \xu\le x\le \xa,\vspace{1ex}\\
	0<(1+\rho x)v'(x) < \al^{-\gam};&\quad x>\xa.
\end{cases}
\end{align}\end{linenomath*}
Under these additional conditions, \eqref{eq:HJB} becomes the following free-boundary problem (FBP):
\begin{linenomath*}\begin{align}\label{eq:FBP}
\begin{cases}
	\displaystyle
	\kap \frac{v'^2(x)}{v''(x)} - \al\left(\frac{x}{\xu}-1\right) v'(x) + \del v(x) = \frac{\al^{1-\gam}}{1-\gam};
	&\quad \xu\le x \le \xa,\vspace{1ex}\\
	\displaystyle
	\kap \frac{v'^2(x)}{v''(x)} - (r+\rho) x v'(x) + \del v(x) = \frac{\gam}{1-\gam}\big((1+\rho x)v'(x)\big)^{1-\frac{1}{\gam}};
	&\quad x > \xa,\vspace{1ex}\\
	\displaystyle
	\lim_{x\to\xu^+}\frac{\vp'(x)}{\vp''(x)} = 0,&\vspace{1ex}\\
	(1+\rho\xa)v'(\xa) = \al^{-\gam},
\end{cases}
\end{align}\end{linenomath*}
in which $\xa\ge \xu$ is unknown, and in which $\kap$ is defined by
\begin{linenomath*}\begin{align}\label{eq:kap}
\kap = \frac{(\mu-r)^2}{2\sig^2}.
\end{align}\end{linenomath*}
In anticipation that $v$ is increasing and concave, we apply the Legendre transform to $v$ to define its convex dual $u$ by
\begin{linenomath*}\begin{align}\label{eq:convexcon2}
u(y) := \sup_{x\ge\xu}\big\{v(x) - xy \big\};\quad y>0.
\end{align}\end{linenomath*}
Here, we assume $\lim_{x\to\xu^+} \vp'(x)=+\infty$, which we show in Proposition \ref{prop:HJB} below. By using the relationships
\begin{linenomath*}\begin{align}\label{eq:Legendre}
v\big(I(y)\big) = u(y) - y u'(y),\quad I(y) = -u'(y),\quad\text{and}\quad v''\big(I(y)\big)=-\frac{1}{u''(y)},
\end{align}\end{linenomath*}
in which $I(\cdot)$ is the inverse of $v'(\cdot)$ (that is, $v'\big(I(y)\big) = y$, for $y>0$), FBP \eqref{eq:FBP} transforms into the following FBP:
\begin{linenomath*}\begin{align}
&\displaystyle\label{eq:FBPDual1}
-\kap y^2 u''(y) + \left(r+\rho(1-\al)-\del\right)y u'(y) + \del u(y) = \frac{\al^{1-\gam}}{1-\gam} -\al y;
\qquad y\ge \ya,\vspace{1ex}\\
&\displaystyle\label{eq:FBPDual2}
-\kap y^2 u''(y) + \left(r+\rho -\del\right)y u'(y) + \del u(y) = \frac{\gam}{1-\gam}\big(y-\rho y u'(y)\big)^{1-\frac{1}{\gam}};
\qquad 0< y < \ya,\vspace{1ex}\\
&\displaystyle\label{eq:FBPDual3}
\lim_{y\to+\infty} u'(y) = -\xu,&\vspace{1ex}\\
&\displaystyle\label{eq:FBPDual4}
\lim_{y\to+\infty}y u''(y) = 0,&\vspace{1ex}
\intertext{and}
\label{eq:FBPDual5}
&\ya-\rho \ya u'(\ya) = \al^{-\gam},
\end{align}\end{linenomath*}
in which $\ya=v'(\xa)$ is unknown.

It is easier to analyze $u$'s second-order differential equation in \eqref{eq:FBPDual1} on $[\ya, +\infty)$ by transforming it into a system of first-order ODEs.  Specifically, by formally defining $\varphi$ and $H$ by
\begin{linenomath*}\begin{align}\label{eq:varphi}
\varphi(y) = y - \rho y u'(y),
\end{align}\end{linenomath*}
and
\begin{equation}\label{eq:H}
H(y) = \dfrac{1}{\varphi(y)} \left[ \del u(y) - \dfrac{\gam}{1 - \gam} \varphi(y)^{1 - \frac{1}{\gam}} - \dfrac{r + \rho - \del}{\rho} \big( \varphi(y) - y \big) \right],
\end{equation}
respectively, and by manipulating these expressions via the differential equation \eqref{eq:FBPDual1} and the free-boundary condition in \eqref{eq:FBPDual5}, we obtain the system in part $(i)$ of the following proposition. (As an aside, we find the value of $H$ at the free-boundary $\ya$ by first solving for $u$ on $(0, \ya)$ and by using continuity of $u$ to obtain $u(\ya)$ and, then, $H(\ya)$.)  Proposition \ref{prop:FBP-sol} provides the complete solution of FBP \eqref{eq:FBPDual1}--\eqref{eq:FBPDual5}.

\begin{proposition}\label{prop:FBP-sol}
Define the constant $\lam\in(-\del/\kap, 0)$ by
\begin{linenomath*}\begin{align}\label{eq:lambda}
	\lam := \frac{1}{2\kap}\left( \big(\kap + r + \rho(1-\al) - \del\big) - \sqrt{\big(\kap + r + \rho(1-\al) - \del\big)^2 + 4\del\kap}\right),
\end{align}\end{linenomath*}
and define constants $0<\eta_1<\eta_2$ by
\begin{linenomath*}\begin{align}\label{eq:eta1eta2}
	\eta_1 :=  \frac{\lam\al^{-\gam}}{(\lam-1)(1+\rho\xu)},
	\quad\text{and}\quad
	\eta_2 :=  \frac{\al^{-\gam}}{1+\rho\xu}.
\end{align}\end{linenomath*}
Then, we have:
\begin{enumerate}
	\item[$(i)$]   There exists a constant $\ya\in(\eta_1, \eta_2)$, an increasing function $\varphi:\left(0,\ya\right]\to (0,\al^{-\gam}]$, and a function $H:\left(0,\ya\right]\to (0,\kap/\rho]$  satisfying the system:
	\begin{linenomath*}\begin{align}\label{eq:phi-H-FBP-sys}
		\begin{cases}\displaystyle
			\varphi'(y) = \frac{\rho}{\kap y}\left(\frac{\kap}{\rho}-H(y)\right)\varphi(y),\vspace{1ex}\\
			\displaystyle
			H'(y) = \frac{\rho}{\kap y}\left(\frac{\kap}{\rho}-H(y)\right)
			\left(\varphi(y)^{-\frac{1}{\gam}} - H(y) - \frac{r+\rho-\del}{\rho}\right)
			+\frac{r+\rho}{\rho\,\varphi(y)}
			- \frac{\del}{\rho y},
			\vspace{1ex}\\
			\varphi(\ya) = \al^{-\gam},\vspace{1ex}\\
			H(\ya) = \displaystyle\frac{\kap}{\rho}\left[1 -\lam \left(1-\frac{\ya}{\eta_1}\right)\right],
		\end{cases}
	\end{align}\end{linenomath*}
	for $0<y\le\ya$. Furthermore, $\lim_{y\to0^+}H(y)\le\kap/\rho$.
	
	\item[$(ii)$] A solution of FBP \eqref{eq:FBPDual1}-\eqref{eq:FBPDual5} is given by $\ya$ as in $(i)$ and by $u:\Rb^+ \to \Rb$ given by
	\begin{linenomath*}\begin{align}\label{eq:u_ge_ys}
		u(y) =
		\begin{cases}
			\displaystyle
			\frac{\ya(1+\rho \xu)-\al^{-\gam}}{\rho\lam} \left(\frac{y}{\ya}\right)^{\lam} -\xu y + \frac{\al^{1-\gam}}{\del(1-\gam)};
			&\quad y\ge\ya,\vspace{1ex}\\
			\displaystyle
			\frac{1}{\del}\left[
			\varphi(y)H(y)
			+ \frac{\gam}{1-\gam}\varphi(y)^{1-\frac{1}{\gam}}
			+ \frac{r+\rho-\del}{\rho}\Big(\varphi(y)-y\Big)
			\right];
			&\quad 0<y<\ya,
		\end{cases}
	\end{align}\end{linenomath*}
	in which $\varphi$ and $H$ are as in $(i)$. Furthermore, $u \in \Cc^2(\R^+)$ is strictly decreasing and convex, and \hfill \break $\lim_{y\to0^+} u'(y)=-\infty$.
\end{enumerate}
\end{proposition}

\begin{proof}
In various parts of this proof, we will use the fact that $\lam$ in \eqref{eq:lambda} solves the quadratic equation
\begin{linenomath*}\begin{align}\label{eq:lam_quadratic_eq}
	f(\lam):=-\kap\lam^2+\big(\kap+r+\rho(1-\al)-\del\big)\lam + \del=0.
\end{align}\end{linenomath*}	
Let $\lam'$ denote the other zero of the quadratic function $f$, which is given by
\begin{linenomath*}\begin{align}\label{eq:lambdap}
	\lam' := \frac{1}{2\kap}\left( \big(\kap + r + \rho(1-\al) - \del \big) + \sqrt{\big(\kap + r + \rho(1-\al) - \del\big)^2 + 4\del\kap}\right)>1.
\end{align}\end{linenomath*}
That $\lam\in(-\del/\kap, 0)$ follows from $f(-\del/\kap)=-\big(r+\rho(1-\al)\big)\del/\kap<0$ and $f(0)=\del>0$. That $\lam'>1$ follows from $f(1)=r+\rho(1-\al)>0$ and $\lim_{\xi\to+\infty} f(\xi) = -\infty$. Below, we prove $(i)$ and then $(ii)$.\vspace{1em}

\noindent \underline{Proof of $(i)$}: When reading this part of the proof, it is helpful to refer to Figure \ref{fig:phi_H} in Section \ref{sec:numerics} for visual reference. Define the set
\begin{linenomath*}\begin{align}\label{eq:DC}
	\Dc:=\left\{(y,\varphi,H): y,\varphi>0, \, 0<H<\frac{\kap}{\rho}\right\},
\end{align}\end{linenomath*}
and functions
\begin{linenomath*}\begin{align}\label{eq:g1}
	g_1(y,\varphi,H) &:= \frac{\rho}{\kap y}\left(\frac{\kap}{\rho}-H\right)\varphi,
	\intertext{and}
	\label{eq:g2}
	g_2(y,\varphi,H) &:= \frac{\rho}{\kap y}
	\left(\frac{\kap}{\rho}-H\right)
	\left(\varphi^{-\frac{1}{\gam}} - H - \frac{r+\rho-\del}{\rho}\right)
	+\frac{r+\rho}{\rho\,\varphi}-\frac{\del}{\rho y},
\end{align}\end{linenomath*}
for $(y,\varphi,H)\in\Dc$. For a constant $\eta\in(\eta_1,\eta_2)$, consider the boundary-value problem
\begin{linenomath*}\begin{align}\label{eq:phiH-eta}
	\begin{cases}\displaystyle
		\varphi'(y)=g_1(y, \varphi(y), H(y)\big)\vspace{1ex},\\
		H'(y)= g_2(y, \varphi(y), H(y)\big)\vspace{1ex},\\
		\varphi(\eta) = \al^{-\gam},\vspace{1ex}\\
		H(\eta) = \displaystyle\frac{\kap}{\rho}\left[1 -\lam \left(1-\frac{\eta}{\eta_1}\right)\right],
	\end{cases}
\end{align}\end{linenomath*}
for $\big(y, \varphi(y), H(y)\big) \in \Dc$. Because $\al > 0$ and $\eta > \eta_1$, the boundary conditions in \eqref{eq:phiH-eta} are inside $\Dc$. Furthermore, $g_1$ and $g_2$ are locally Lipschitz continuous with respect to $\varphi$ and $H$ in $\Dc$, since they are only unbounded (or have unbounded partial derivatives) if $y=0$ or $\varphi=0$. It, then, follows that \eqref{eq:phiH-eta} has a unique solution that extends to the boundary of $\Dc$. Denote this solution by $(\varphi_\eta(\cdot), H_\eta(\cdot)\big):(\eps(\eta),\eta]\to\Rb$ for some constant $\eps(\eta)\in[0, \eta)$ such that $(\eps(\eta),\eta]$ is the maximal domain over which the solution exists (within $\Dc$). We prove additional properties of $(\varphi_\eta(\cdot), H_\eta(\cdot)\big)$ for $\eta\in(\eta_1,\eta_2)$ in Lemma \ref{lem:phiH-eta} in Appendix \ref{app:aux_constrained} and use those properties in the rest of this proof.

Note that, because the solution $(\varphi_\eta(\cdot), H_\eta(\cdot)\big)$ continuously depends on $\eta$ because of the aforementioned local Lipschitz property of $g_1$ and $g_2$, the mapping $\eta\mapsto \eps(\eta)$ is continuous for $\eta\in(\eta_1,\eta_2)$. Our goal is to show that there exists a constant $\ya\in(\eta_1,\eta_2)$ such that $\eps(\ya)=0$; that is, the solution $(\varphi_{\ya}(\cdot), H_{\ya}(\cdot)\big)$ is defined over the interval $(0,\ya]$.

To show the existence of such $\ya$, we first show that for every $y'\in(0,\eta_1)$, there exists a constant $\eta\in(\eta_1,\eta_2)$ such that the solution $(\varphi_\eta(\cdot), H_\eta(\cdot)\big)$ exists through a point $(y', \varphi', \kap/\rho)\in\overline{\Dc}_1$, in which we have defined
\begin{linenomath*}\begin{align}\label{eq:Dc1}
	\overline{\Dc}_1 := \big\{(y,\varphi,\kap/\rho):y\in(0,\eta_1), \varphi\in(0,\al^{-\gam}) \big\}.
\end{align}\end{linenomath*}
To prove this statement, let $\Bc$ be the set of all $y'\in(0,\eta_1)$ such that there exists a solution $(\varphi_\eta(\cdot), H_\eta(\cdot)\big)$, $\eta\in(\eta_1,\eta_2)$, which exits through a point $(y', \varphi', \kap/\rho)\in\overline{\Dc}_1$. We want to show that $\Bc=(0,\eta_1)$. By Lemma \ref{lem:phiH-eta}.$(ii)$, $\Bc$ is nonempty. 
From Lemma \ref{lem:phiH-eta}.$(iv)$ and the continuity of $(\varphi_\eta(\cdot), H_\eta(\cdot)\big)$ with respect to $\eta$, it follows that if $y_0\in \Bc$, then $y\in \Bc$ for all $y\in(y_0,\eta_1)$. Therefore, we have one of the following scenarios: (a) $\Bc=(\yt,\eta_1)$ for some $\yt\in(0,\eta_1)$, (b) $\Bc=[\yt,\eta_1)$ for some $\yt\in(0,\eta_1)$, or (c) $\Bc=(0,\eta_1)$.

In scenario (a), there exists a monotone increasing sequence $\{\xi_n\}_{n=1}^\infty$ in $(\eta_1,\eta_2)$, such that solutions  $\big(\varphi_{\xi_n}(\cdot), H_{\xi_n}(\cdot)\big)$ all exit from $\overline{\Dc}_1$ and $\lim_{n\to+\infty}\eps(\xi_n)=\yt$.  By Lemma \ref{lem:phiH-eta}.$(iii)$-$(iv)$, we must have $\xi_n<\eta_2-\epsilon$ for some $\epsilon>0$ and for all $n$. Thus, $\lim_{n\to+\infty}\xi_n = \xi_{\infty}$ for some constant $\xi_\infty\in(\eta_1,\eta_2-\eps]$.  Furthermore, by continuity of $(\varphi_\eta(\cdot), H_\eta(\cdot)\big)$ with respect to $\eta$, we must have that $(\varphi_{\xi_\infty}(\cdot), H_{\xi_\infty}(\cdot)\big)$ exits $\Dc$ through some point $(\yt, \pht, \kap/\rho)\in\overline{\Dc}_1$. This implies $\yt\in\Bc$, which contradicts the assumption that $\yt\notin\Bc$. Thus, scenario (a) is impossible.

In scenario (b), because $\yt \in \Bc$, it follows from Lemma \ref{lem:phiH-eta}.$(iii)$ that there exists a constant $\etat\in(\eta_1,\eta_2)$ such that $(\varphi_{\etat}(\cdot), H_{\etat}(\cdot)\big)$ exits $\Dc$ through a point $(\yt, \pht, \kap/\rho)\in\overline{\Dc}_1$. Since $\etat<\eta_2$ and $\yt>0$, from continuity of $(\varphi_\eta(\cdot), H_\eta(\cdot)\big)$ with respect to $\eta$, it follows that for some $y'\in(0,\yt)$, there exists an $\eta'\in(\etat,\eta_2)$ such that $(\varphi_{\eta'}(\cdot), H_{\eta'}(\cdot)\big)$ exits $\Dc$ through a point $(y', \varphi', \kap/\rho)\in\overline{\Dc}_1$. In other words, $y'\in\Bc$, which contradicts $\yt=\min \Bc$. Thus, scenario (b) is also impossible. We conclude that the only possible scenario is (c), in other words, $\Bc=(0,\eta_1)$.

Finally, define $\ya=\inf\big\{\eta\in(\eta_1,\eta_2): \eps(\eta)\in \Bc\big\}$. From Lemma \ref{lem:phiH-eta}.$(ii)$-$(iii)$, we must have $\ya\in(\eta_1,\eta_2)$. From continuity of $(\varphi_\eta(\cdot), H_\eta(\cdot)\big)$ with respect to $\eta$, we deduce $\eps(\ya)=0$, and
\begin{linenomath*}\begin{align}
	\lim_{y\to0^+}H_{\ya}(y)\le \frac{\kap}{\rho}.
\end{align}\end{linenomath*}
Thus, the solution $(\varphi_{\ya}(\cdot), H_{\ya}(\cdot)\big)$ satisfies \eqref{eq:phi-H-FBP-sys} for $y\in(0,\ya)$. Finally, that $\varphi_{\ya}(\cdot)$ is increasing follows from
\begin{linenomath*}\begin{align}
	\varphi'(y) = \frac{\rho}{\kap y}\left(\frac{\kap}{\rho}-H(y)\right)\varphi(y)>0,
\end{align}\end{linenomath*}
for all $y\in(0,\ya)$, since $\big(y, \varphi_{\ya}(y), H_{\ya}(y)\big)\in\Dc$.
\vspace{1em}

\noindent \underline{Proof of $(ii)$}:  The solution of the Euler equation \eqref{eq:FBPDual1} is
\begin{linenomath*}\begin{align}\label{eq:Euler}
	u(y) = C y^\lam + C' y^{\lam'} -\xu y + \frac{\al^{1-\gam}}{\del(1-\gam)};\qquad y\ge \ya,
\end{align}\end{linenomath*}
in which $C$ and $C'$ are constants to be determined, and $\lam\in(-\del/\kap, 0)$ and $\lam'\in(1,+\infty)$ are given by \eqref{eq:lambda} and \eqref{eq:lambdap}, respectively. By \eqref{eq:Euler}, conditions \eqref{eq:FBPDual3} and \eqref{eq:FBPDual4} become
\begin{linenomath*}\begin{align}
	\begin{cases}\displaystyle
		\lim_{y\to+\infty}\left(\lam C y^{\lam-1}+ \lam' C' y^{\lam'-1} \right) = 0,\\
		\displaystyle
		\lim_{y\to+\infty}\left(\lam(1-\lam) C y^{\lam-1}+ \lam'(1-\lam') C' y^{\lam'-1}\right)=0.
	\end{cases}
\end{align}\end{linenomath*}
Since $\lam<0$ and $\lam'>1$, the system above can only hold if $C'=0$. So, we must have,
\begin{linenomath*}\begin{align}
	u(y) = C y^{\lam} -\xu y + \frac{\al^{1-\gam}}{\del(1-\gam)};\qquad y\ge \ya.
\end{align}\end{linenomath*}
From \eqref{eq:FBPDual5}, we obtain $C = \frac{\ya(1+\rho \xu)-\al^{-\gam}}{\rho\lam(\ya)^{\lam}}$, which yields
\begin{linenomath*}\begin{align}\label{eq:u_ge_ys2}
	u(y) = \frac{\ya(1+\rho \xu)-\al^{-\gam}}{\rho\lam} \left(\frac{y}{\ya}\right)^{\lam} -\xu y + \frac{\al^{1-\gam}}{\del(1-\gam)};\qquad y\ge \ya.
\end{align}\end{linenomath*}
Thus, FBP \eqref{eq:FBPDual1}--\eqref{eq:FBPDual5} reduces to the following FBP:
\begin{linenomath*}\begin{align}\label{eq:u-FBP-Risky}
	\begin{cases}
		-\kap y^2 u''(y) + \left(r+\rho -\del\right)y u'(y) + \del u(y) = \frac{\gam}{1-\gam}\big(y-\rho y u'(y)\big)^{1-\frac{1}{\gam}};
		\qquad 0< y < \ya,\\
		\ya-\rho \ya u'(\ya) = \al^{-\gam},\\
		u(\ya)  = 
		\left(\frac{1+\rho\xu}{\rho\lam}-\xu\right)\ya + \left(\frac{\al}{\del(1-\gam)} - \frac{1}{\rho\lam}\right)\al^{-\gam}.
	\end{cases}
\end{align}\end{linenomath*}

Now, let $\varphi$, $H$, and $\ya$ be as determined in part $(i)$ of this proposition; then, we claim that $u$ defined by
\begin{linenomath*}\begin{align}\label{eq:u_given_H}
	u(y) = \frac{1}{\del}\left[
	\varphi(y)H(y)
	+ \frac{\gam}{1-\gam}\varphi(y)^{1-\frac{1}{\gam}}
	+ \frac{r+\rho-\del}{\rho}\big(\varphi(y)-y\big)
	\right];\quad y\in(0, \ya],
\end{align}\end{linenomath*}
satisfies FBP \eqref{eq:u-FBP-Risky}.  Indeed, because $\varphi(\ya) = \al^{-\gam}$ and $H(\ya) = \frac{\kap}{\rho}\left[1 -\lam \left(1-\frac{\ya}{\eta_1}\right)\right]$, one can show that $u$ in \eqref{eq:u_given_H} satisfies the second free-boundary condition in \eqref{eq:u-FBP-Risky}.  Next, if we differentiate $u$ twice, substitute for $\varphi'$ and $H'$ from \eqref{eq:phi-H-FBP-sys} each time, then we obtain
\begin{equation}\label{eq:u_prime}
	u'(y) = \frac{1}{\rho}-\frac{\varphi(y)}{\rho y},
\end{equation}
and
\begin{equation}\label{eq:u_prime2}
	u''(y) = \dfrac{\varphi(y)H(y)}{\kap y^2},
\end{equation}
for $0<y<\ya$. Note that \eqref{eq:u_prime} and $\varphi(\ya) = \al^{-\gam}$ give us the first free-boundary condition in \eqref{eq:u-FBP-Risky}.  If we substitute for $u'$ and $u''$ from \eqref{eq:u_prime} and \eqref{eq:u_prime2}, respectively, in the non-linear differential equation in \eqref{eq:u-FBP-Risky}, then we obtain
\begin{linenomath*}\begin{align}\label{eq:u-phi-FBP-sys}
	&-\kap y^2 u''(y) + \left(r+\rho -\del\right)y u'(y) + \del u(y) - \frac{\gam}{1-\gam}\big(y-\rho y u'(y)\big)^{1-\frac{1}{\gam}} \notag \\
	&= - \varphi(y)H(y) + \dfrac{r + \rho - \del}{\rho}\big(y - \varphi(y) \big) + \del u(y) - \frac{\gam}{1-\gam} \varphi(y)^{1 - \frac{1}{\gam}} = 0,
\end{align}\end{linenomath*}
in which the last equality follows from the definition of $u$ in \eqref{eq:u_given_H}.  We have, thereby, shown that $\ya$ from part $(i)$ and $u$ given by \eqref{eq:u_ge_ys} solve FBP \eqref{eq:FBPDual1}-\eqref{eq:FBPDual5}.

Next, we show that $u$ given by \eqref{eq:u_ge_ys} is decreasing and convex; note that $u \in \Cc^2(\R^+)$ is continuously twice differentiable by construction.  For $y \ge \ya$, these properties of $u$ directly follow by differentiating \eqref{eq:u_ge_ys2} as follows:
\begin{linenomath*}\begin{align}\label{eq:up_large_y}
	u'(y) = \frac{\ya(1+\rho \xu)-\al^{-\gam}}{\rho y^{*\lam}}y^{\lam-1} - \xu<0,
\end{align}\end{linenomath*}
and
\begin{linenomath*}\begin{align}\label{eq:upp_large_y}
	u''(y) = \frac{(\lam-1)\big(\ya(1+\rho \xu)-\al^{-\gam}\big)}{\rho y^{*\lam}}y^{\lam-2}>0,
\end{align}\end{linenomath*}
for $y\ge \ya$, in which, to get the inequalities, we used $\lam<0$ and $\ya<\eta_2=\frac{\al^{-\gam}}{1+\rho\xu}$, which we proved earlier.  That $u$ is convex on $(0, \ya)$ follows from \eqref{eq:u_prime2}, $\varphi > 0$, and $H > 0$; we proved the latter two inequalities in part $(i)$.  Also, \eqref{eq:up_large_y} implies $u'(\ya) < 0$, and $u$ convex on $(0, \ya)$ implies $u'(y) < 0$ for all $y \in (0, \ya)$.

It only remains to show that $\lim_{y\to0^+} u'(y)=-\infty$. Suppose, on the contrary, that $\lim_{y\to0^+} u'(y)\ne -\infty$. Because $u'$ is increasing and $\lim_{y\to+\infty}u'(y)=-\xu$, we must have $\lim_{y\to0^+} u'(y) = M$ for some constant $M<-\xu<0$. From \eqref{eq:u_prime}, we have $\varphi(y)/y = 1-\rho u'(y)$ for $0<y<\ya$. Therefore,
\begin{linenomath*}\begin{align}\label{eq:contradiction}
	\lim_{y\to 0^+}\frac{\varphi(y)}{y} = 1-\rho M >1.
\end{align}\end{linenomath*}
The above limit implies $\lim_{y\to0^+}\varphi(y)=0$. By L'H\^{o}pital's rule, \eqref{eq:phi-H-FBP-sys}, and $\lim_{y\to0^+}H(y)\in (0, \kap/\rho]$ from part $(i)$, we obtain a contradiction
\begin{linenomath*}\begin{align}
	\lim_{y\to 0^+}\frac{\varphi(y)}{y} = \lim_{y\to 0^+}\varphi'(y)
	=\lim_{y\to 0^+} \frac{\varphi(y)}{y}\left(1-\frac{\rho H(y)}{\kap}\right)
	< \lim_{y\to 0^+}\frac{\varphi(y)}{y},
\end{align}\end{linenomath*}
Thus, we must have $\lim_{y\to0^+} u'(y)=-\infty$.
\end{proof}

\begin{remark}
It is possible to find a differential equation for $\varphi$ of Proposition \ref{prop:FBP-sol}.$(i)$ that does not involve $H$. Indeed, substituting $u'=\frac{1}{\rho}-\frac{\varphi}{\rho y}$ and $u''=\frac{\varphi}{\rho y^2}-\frac{\varphi'}{\rho y}$ into \eqref{eq:FBPDual2} yields
\begin{linenomath*}\begin{align}\label{eq:u-phi-sys1}
	\del u(y) = -\frac{\kap y}{\rho} \varphi'(y)+\frac{\kap+r+\rho-\del}{\rho}\,\varphi(y)+\frac{\gam}{1-\gam}\varphi(y)^{1-\frac{1}{\gam}}-\frac{r+\rho-\del}{\rho}\,y;\quad 0<y\le \ya.
\end{align}\end{linenomath*}
By differentiating this equation and substituting $u'=\frac{1}{\rho}-\frac{\varphi}{\rho y}$, we obtain the following second-order differential equation for $\varphi$:
\begin{linenomath*}\begin{align}\label{eq:phi-TVP}
	\frac{\kap}{\rho}y^2\varphi''(y) + \left(\varphi(y)^{-\frac{1}{\gam}}-\frac{r+\rho-\del}{\rho}\right)y\varphi'(y) - \frac{\del}{\rho}\varphi(y) + \frac{r+\rho}{\rho}\,y=0;
	\quad 0<y\le\ya.
\end{align}\end{linenomath*}
The above equation provides a link between Proposition \ref{prop:FBP-sol} and Propositions 3.1 and 3.2 in \cite{AngoshtariBayraktarYoung2020}. Indeed, by setting $\kap=0$, the above equation reduces to the differential equation in (3.17) and (3.29) of \cite{AngoshtariBayraktarYoung2020}. Thus, as $\kap\to0^+$, $\varphi$ of Proposition \ref{prop:FBP-sol} becomes $\psi$ of Propositions 3.1 and 3.2 in \cite{AngoshtariBayraktarYoung2020}. This relationship is expected because, by letting $\kap\to 0^+$, the risky asset becomes redundant and the optimal policy only invests in the riskless asset, which is the scenario analyzed in \cite{AngoshtariBayraktarYoung2020}. \qed \vspace{1em}
\end{remark}

Proposition \ref{prop:FBP-sol} provides a strictly decreasing and convex function $u$ and corresponding free boundary $\ya$ that solve \eqref{eq:FBPDual1}--\eqref{eq:FBPDual5}. By reversing the Legendre transform \eqref{eq:convexcon2}, we obtain an increasing and concave solution of FBP \eqref{eq:FBP}. We prove this result in the following proposition.

\begin{proposition}\label{prop:HJB}
Let $\lam$, $\ya$, $\varphi(y)$, $H(y)$, and $u(y)$ be as in Proposition \ref{prop:FBP-sol}, and let $J(\xi):(-\infty, -\xu)\to(0,+\infty)$ be the inverse of $u'(y)$, that is, $u'\big(J(\xi)\big)=\xi$ for $\xi<-\xu$. Define
\begin{linenomath*}\begin{align}\label{eq:xa}
	\xa &:= -u'(\ya) = \frac{\al^{-\gam}}{\rho\ya}-\frac{1}{\rho},\\
	\vpt(x) &:= u\big(J(-x)\big) + x J(-x);\quad x>\xu,
\end{align}\end{linenomath*}
\begin{linenomath*}\begin{align}\label{eq:CS-sol}
	\cs(x) :=
	\begin{cases}
		\al; &\quad \xu\le x\le \xa,\vspace{1ex}\\
		\displaystyle
		\Big(\varphi\big(J(-x)\big)\Big)^{-\frac{1}{\gam}};&\quad x>\xa,
	\end{cases}
\end{align}\end{linenomath*}
and
\begin{linenomath*}\begin{align}\label{eq:thts}
	\thts(x) :=
	\begin{cases}
		\dfrac{(\mu-r)(1-\lam)}{\sig^2} \, (x-\xu);
		&\quad \xu\le x\le \xa,\vspace{1ex}\\
		\displaystyle
		\frac{\mu-r}{\kap\sig^2} \, H\big(J(-x)\big)(1+\rho x);
		&\quad x>\xa.
	\end{cases}
\end{align}\end{linenomath*}		
Then, $\xa$, $\vpt(x)$, $\thts(x)$, and $\cs(x)$ satisfy \eqref{eq:cs-cand}, \eqref{eq:thts-cand},  \eqref{eq:VE}, and \eqref{eq:FBP}. Furthermore, $\vpt \in \Cc^2\big([\xu, +\infty) \big)$ is strictly increasing and concave, $\xa>\xu$, and we can write $v$ as follows:
\begin{equation}\label{eq:v_explicit}
	v(x) = 
	\begin{cases}
		\left( \dfrac{\rho \ya}{\al^{-\gam} - \ya(1 + \rho \xu)} \right)^{\frac{1}{\gam - 1}} \left\{ - \dfrac{\ya}{\lam} (x - \xu)^{\frac{\lam}{\lam - 1}} + (x - \xu) \right\} + \dfrac{\al^{1 - \gam}}{\del(1 - \gam)}; &\quad \xu \le x \le \xa, \vspace{0.8em} \\
		\dfrac{1}{\del}\left[\varphi\big(J(-x)\big)H\big(J(-x)\big) + \dfrac{\gam}{1-\gam}\varphi\big(J(-x)\big)^{1-\frac{1}{\gam}} + \dfrac{r+\rho}{\rho}\big(\varphi\big(J(-x)\big)-J(-x)\big) \right]; &\quad x > \xa.
	\end{cases}
\end{equation}
In particular, 
the expression for $v$ in \eqref{eq:v_explicit} implies that $\lim_{x \to \xu^+} v'(x) = +\infty$.
\end{proposition}

\begin{proof}
By Proposition \ref{prop:FBP-sol}, $u':(0,+\infty) \to (-\infty, -\xu)$ is an increasing function such that $\lim_{y\to0^+}u'(y) = -\infty$ and $\lim_{y\to+\infty} u'(y) =-\xu$. Therefore, its inverse $J:(-\infty,-\xu)\to (0,+\infty)$ is an increasing function such that $\lim_{\xi\to-\xu^{-}}J(\xi)=+\infty$ and $\lim_{\xi\to-\infty} J(\xi)=0$.

The expression for $\xa$ follows from \eqref{eq:u_prime}, the expression for $v$ follow \eqref{eq:Legendre}, and the expression for $\cs$ follows from \eqref{eq:cs-cand}, \eqref{eq:Legendre}, and \eqref{eq:u_prime}. To obtain \eqref{eq:thts}, use \eqref{eq:thts-cand} and \eqref{eq:Legendre} to obtain
\begin{linenomath*}\begin{align}\label{eq:thts2}
	\thts(x) &:= - \, \frac{\mu-r}{\sigma^2} \frac{v'(x)}{v''(x)} = \frac{\mu-r}{\sigma^2} J(-x)u''\big(J(-x)\big);\quad x>\xu.
\end{align}\end{linenomath*}
We consider two cases: $x\in[\xu,\xa]$ and $x>\xa$. For the former case, we argue as follows. By \eqref{eq:up_large_y}, $u'(y)\in(-\xa,-\xu)$ for $y>\ya$ and, therefore, $J(\xi)>\ya$ for $\xi\in(-\xa,-\xu)$. It then follows from \eqref{eq:up_large_y} that
\begin{linenomath*}\begin{align}
	\xi=u'\big(J(\xi)\big)=\frac{\ya(1+\rho \xu)-\al^{-\gam}}{\rho y^{*\lam}}J(\xi)^{\lam-1} - \xu,
	\quad\Longrightarrow\quad
	J(\xi)^{\lam-1} = \left(\frac{\rho y^{*\lam}(\xu+\xi)}{\ya(1+\rho \xu)-\al^{-\gam}}\right),
\end{align}\end{linenomath*}
for $\xi\in(-\xa,-\xu)$. By using \eqref{eq:upp_large_y} and \eqref{eq:thts}, we obtain
\begin{linenomath*}\begin{align}
	\thts(x) &= \frac{\mu-r}{\sigma^2} J(-x)u''\big(J(-x)\big)
	= \frac{\mu-r}{\sigma^2} \frac{(\lam-1)\big(\ya(1+\rho \xu)-\al^{-\gam}\big)}{\rho y^{*\lam}}J(-x)^{\lam-1}
	=\frac{(\mu-r)(1-\lam)}{\sig^2}(x-\xu),
\end{align}\end{linenomath*}
for $x\in[\xu,\xa]$. To obtain \eqref{eq:thts} for $x>\xa$, note that by the definition of $J$ and \eqref{eq:u_prime}, we have
\begin{linenomath*}\begin{align}
	-x = u'\big(J(-x)\big) = \frac{1}{\rho}\left(1-\frac{\varphi\big(J(-x)\big)}{J(-x)}\right)
	\quad\Longrightarrow\quad
	\frac{\varphi\big(J(-x)\big)}{J(-x)} = 1 + \rho x,
\end{align}\end{linenomath*}
for $x>\xa$. From \eqref{eq:u_prime2}, it follows that
\begin{linenomath*}\begin{align}
	J(-x)u''\big(J(-x)\big) = \frac{1}{\kap}H\big(J(-x)\big)(1+\rho x),
\end{align}\end{linenomath*}
for $x>\xa$. By substituting for $J(-x)u''\big(J(-x)\big)$ in \eqref{eq:thts2}, we obtain \eqref{eq:thts} for $x>\xa$. We can double check that $\thts(x)$ is continuous at $x=x^*$ as follows:
\begin{linenomath*}\begin{align}
	&\frac{1}{\kap}H\big(J(-\xa)\big)(1+\rho\xa)
	= \frac{1}{\kap}H\big(\ya\big)(1+\rho\xa)
	= \frac{1}{\rho}\left(1 -\lam + \frac{\lam\ya}{\eta_1}\right)(1+\rho\xa)\\
	&= \frac{1}{\rho}\left(1 -\lam + (\lam-1)\frac{1+\rho\xu}{1+\rho\xa}\right)(1+\rho\xa)
	= \frac{1}{\rho}(1-\lam)\left(1 -\frac{1+\rho\xu}{1+\rho\xa}\right)(1+\rho\xa)
	= (1-\lam)(\xa-\xu),
\end{align}\end{linenomath*}
in which we used $u'(\ya)=-\xa$ to get the first equality and the second terminal condition in \eqref{eq:phi-H-FBP-sys} for the second equality. To get the third equality, we used the boundary condition $(1+\rho\xa)v'(\xa) = \al^{-\gam}$ in \eqref{eq:FBP} and the definition of $\eta_1$ in \eqref{eq:eta1eta2} to obtain
\begin{linenomath*}\begin{align}
	(1+\rho\xa)\ya = \al^{-\gam} = \frac{\lam-1}{\lam}(1+\rho\xu)\eta_1
	\quad\Longrightarrow\quad
	\frac{\lam\ya}{\eta_1} = (\lam-1)\frac{1+\rho\xu}{1+\rho\xa}.
\end{align}\end{linenomath*}

It is, then, straightforward to show that $\xa$, $\vpt(\cdot)$, $\thts(\cdot)$, and $\cs(\cdot)$ satisfy \eqref{eq:cs-cand}, \eqref{eq:thts-cand},  \eqref{eq:VE}, and \eqref{eq:FBP} by reversing the transformation \eqref{eq:convexcon2} and by using the fact that $\ya$ and $u(\cdot)$ solve FBP \eqref{eq:FBPDual1}--\eqref{eq:FBPDual5}. That $\vpt(\cdot)$ is increasing and strictly concave follows from \eqref{eq:Legendre} since $u(\cdot)$ is decreasing and strictly convex as established by Proposition \ref{prop:FBP-sol}.$(ii)$. Furthermore,
\begin{linenomath*}\begin{align}
	\xa = \frac{\al^{-\gam}}{\rho\ya}-\frac{1}{\rho}
	> \frac{\al^{-\gam}}{\rho\eta_2}-\frac{1}{\rho}=\xu,
\end{align}\end{linenomath*}
because $0 < \ya<\eta_2=\frac{\al^{-\gam}}{1+\rho\xu}$ by Proposition \ref{prop:FBP-sol}.$(i)$.  Finally, the expression for $v$ in \eqref{eq:v_explicit} follows from $\vpt(x) = u\big(J(-x)\big) + x J(-x)$ and the expressions in Proposition \ref{prop:FBP-sol}.
\end{proof}

The next theorem is the main result of the paper and provides the solution of the stochastic control problem \eqref{eq:VF}.

\begin{theorem}\label{thm:VF}
Let $\xa$, $\vpt(x)$, $\thts(x)$, and $\cs(x)$ be as in Proposition \ref{prop:HJB}; then, $V(x,\al) = \vpt(x)$ for all $x \ge \xu$. Furthermore, the optimal investment-to-habit and consumption-to-habit processes are given by $\theta^*_t := \thts(X^*_t)$ and $c^*_t:=\cs(X^*_t)$, respectively, for all $t \ge 0$, in which $(X^*_t)_{t\ge0}$ solves the stochastic differential equation
\begin{linenomath*}\begin{align}\label{eq:XStar}
	\begin{cases}
		\dd X^*_t = \big((\rho+r)X^*_t + (\mu-r)\thts(X^*_t) - (1+\rho X^*_t)\cs(X^*_t)\big)\dd t + \sig\thts(X^*_t) \dd B_t;\quad t\ge0,\\
		X_0=x.
	\end{cases}
\end{align}\end{linenomath*}
\end{theorem}

\begin{proof}
It suffices to show that $\vpt$, $\thts$, and $\cs$ satisfy conditions $(i)$--$(v)$ of Theorem \ref{thm:verif}. Conditions $(i)$, $(ii)$, and $(iv)$ directly follow from Proposition \ref{prop:HJB}. Below, we prove conditions $(iii)$ and $(v)$ of that theorem.\vspace{1em}

\noindent\textbf{Condition $(iii)$:} Let $(X_t)_{t\ge0}$ be an admissible wealth-to-habit process corresponding to a relative investment and consumption policy $(\theta_t,c_t)_{t\ge0}\in\Ac(\al)$. By Proposition \ref{prop:HJB}, $\vpt$ is increasing and $\vp(\xu)=\frac{\al^{1-\gam}}{\del(1-\gam)}$; therefore,
\begin{linenomath*}\begin{align}\label{eq:transversality_lb}
	\ee^{-\del T} \frac{\al^{1-\gam}}{\del(1-\gam)} \le \Eb_{x}\big(\ee^{-\del T} \vp(X_T)\big),
\end{align}\end{linenomath*}
for all $T \ge 0$.
Define the non-negative process $(Y_t)_{t\ge0}$ by
\begin{linenomath*}\begin{align}\label{eq:Y-SDE}
	\begin{cases}
		\dd Y_t = -\big(r+\rho(1-c_t)\big)Y_t\dd t - \dfrac{\mu-r}{\sig} Y_t \dd B_t;\quad t\ge0,\\
		Y_0 = 1.
	\end{cases}
\end{align}\end{linenomath*}
From \eqref{eq:X-SDE}, it follows that
\begin{linenomath*}\begin{align}
	X_t Y_t + \int_0^t c_s Y_s \dd s = x + \int_0^t \left(\sig \theta_s - \frac{\mu-r}{\sig} X_s\right) Y_s \dd B_s,
\end{align}\end{linenomath*}
for any $t>0$. In particular, $\big(X_t Y_t + \int_0^t c_s Y_s \dd s\big)_{t\ge0}$ is a non-negative local martingale and, hence, a supermartingale. Therefore, $\Eb_x\big(X_t Y_t + \int_0^t c_s Y_s \dd s\big) \le x$ which, in turn, yields
\begin{linenomath*}\begin{align}\label{eq:budget}
	0<\Eb_x(X_t Y_t) \le x;\quad t\ge0,
\end{align}\end{linenomath*}
because $X_t,Y_t, c_t>0$, $\Pb$-a.s., for all $t \ge 0$.

Let $u$ be as in Proposition \ref{prop:FBP-sol}.$(ii)$. From \eqref{eq:convexcon2}, we obtain
\begin{linenomath*}\begin{align}\label{eq:v_u_inequality}
	\Eb_{x}\big(\vp(X_T)\big)
	= \Eb_{x}\big( \vp(X_T) - X_T Y_T + X_T Y_T\big)
	\le \Eb_{x}\big( u(Y_T) + X_T Y_T\big)
	\le \Eb\big( u(Y_T)\big) + x,
\end{align}\end{linenomath*}
for all $T>0$, in which we used \eqref{eq:budget} to get the last inequality.

For $0<y<\ya$, \eqref{eq:u_ge_ys} yields
\begin{linenomath*}\begin{align}
	u(y) &= \frac{1}{\del}\left[
	\varphi(y)H(y)
	+ \frac{\gam}{1-\gam}\varphi(y)^{1-\frac{1}{\gam}}
	+ \frac{r+\rho-\del}{\rho}\big(\varphi(y)-y\big)
	\right]\\
	\label{eq:u_upperbound}
	&\le \frac{\gam}{\del(1-\gam)}\varphi(y)^{1-\frac{1}{\gam}}
	+ \frac{(\kap+r+\rho-\del)\al^{-\gam}}{\del\rho}\\
	&\le \frac{(\kap+r+\rho-\del)\al^{-\gam}}{\del\rho},
\end{align}\end{linenomath*}
because $\gam>1$, $\varphi(y)\in(0,\al^{-\gam})$, and $H(y)\in(0,\frac{\kap}{\rho})$ by Proposition \ref{prop:FBP-sol}.$(i)$.
Because $u$ is decreasing by Proposition \ref{prop:FBP-sol}.$(ii)$, we have $u(y) \le \frac{(\kap+r+\rho-\del)\al^{-\gam}}{\del\rho}$ for all $y>0$. Inequalities \eqref{eq:transversality_lb} and \eqref{eq:v_u_inequality}, then, yield
\begin{linenomath*}\begin{align}
	&\ee^{-\del T} \frac{\al^{1-\gam}}{\del(1-\gam)} \le \Eb_{x}\big(\ee^{-\del T} \vp(X_T)\big)
	\le  \ee^{-\del T}\left( x + \Eb\big( u(Y_T)\big)\right)
	\le  \ee^{-\del T}\left( x + \frac{(\kap+r+\rho-\del)\al^{-\gam}}{\del\rho}\right).
\end{align}\end{linenomath*}
Condition $(iii)$ of Theorem \ref{thm:verif} follows by taking the limit as $T\to+\infty$.
\vspace{1em}

\noindent\textbf{Condition $(v)$:}
It suffices to show that \eqref{eq:XStar} has a unique strong solution $(X^*_t)_{t\ge0}$ taking values in the open interval $I:=(\xu, +\infty)$. For $x\in I$, let
\begin{linenomath*}\begin{align}\label{eq:ab}
	b(x) := (r+\rho)x + (\mu-r)\thts(x)-(1+\rho x)\cs(x),
	\quad\text{and} \quad
	a(x):=\sig\thts(x),
\end{align}\end{linenomath*}
be the drift and diffusion terms of \eqref{eq:XStar}, respectively. Note that the drift function $b(x)$ in \eqref{eq:ab} is not globally Lipschitz because of the term $x\cs(x)$. Therefore, standard existence results, such as Theorem 5.2.9 on page 289 of \cite{KS1991}, are not directly applicable here.

Since $b(x)$ and $a(x)$ are locally Lipschitz for $x\in I$, a standard localization argument yields that \eqref{eq:XStar} has a unique strong solution up to an explosion time. In the remaining part of the proof, we show that \eqref{eq:XStar} does not have an exploding solution (that is, a solution that exits $I$ in finite time). For $x\in I$, define
\begin{linenomath*}\begin{align}\label{eq:Feller_fn}
	\psi(x) = \int_{\xa}^x\int_{\xa}^y \frac{2}{a(z)^2}\exp\left(-2\int_z^y \frac{b(\eta)}{a(\eta)^2}\dd\eta\right) \dd z\dd y,
\end{align}\end{linenomath*}
By Feller's test for explosions (see, for example, Theorem 5.5.29 on page 348 of \cite{KS1991}), \eqref{eq:XStar} does not have an exploding solution if $\lim_{x\to+\infty} \psi(x) = \lim_{x\to\xu^+} \psi(x) = +\infty$, which we show next.

For $x\in(\xu,\xa)$, \eqref{eq:CS-sol} and \eqref{eq:thts} yield that $b(x)=(x-\xu)b_0$ and $a(x)=(x-\xu)a_0$, in which $b_0:=r+\rho(1-\al)+\left(\frac{\mu-r}{\sig}\right)^2(1-\lam)>0$ and $a_0:=(\mu-r)(1-\lam)/\sig>0$. It then follows that,
\begin{linenomath*}\begin{align}
	\psi(x) = \frac{2}{a_0^2+b_0}\left\{
	\frac{a_0^2}{a_0^2+b_0}\left[
	\left(\frac{x-\xu}{\xa-\xu}\right)^{1+\frac{b_0}{a_0^2}}
	-1
	\right]
	+\ln\left(\frac{\xa-\xu}{x-\xu}\right)
	\right\},
\end{align}\end{linenomath*}
for $x\in(\xu,\xa)$, which yields that $\lim_{x\to\xu^+} \psi(x) = +\infty$.

It only remains to show that $\lim_{x\to+\infty} \psi(x) = +\infty$. By \eqref{eq:CS-sol} and \eqref{eq:thts}, we have $\cs(x)=\big(\varphi\big(J(-x)\big)\big)^{-\frac{1}{\gam}}$ and $\thts(x)=\frac{\mu-r}{\kap\sig^2}H\big(J(-x)\big)(1+\rho x)$, for $x>\xa$. Furthermore, by the proof of Proposition \ref{prop:FBP-sol}, there exists a constant $H_0$ such that
\begin{linenomath*}\begin{align}\label{eq:H_inequality}
	0<H_0\le H\big(J(-x)\big)\le \frac{\kap}{\rho},
\end{align}\end{linenomath*}
for $x> \xa$. For $y>z>\xa$, we, then, have
\begin{linenomath*}\begin{align}
	\int_z^y \frac{b(\eta)}{a(\eta)^2}\dd\eta
	&= \int_z^y \left(\frac{r+\rho}{\sig^2}\frac{x}{\thts(x)^2} 
	+ \frac{\mu-r}{\sig^2}\frac{1}{\thts(x)}
	-\frac{(1+\rho x)\cs(x)}{\thts(x)^2}\right)\dd\eta\\
	&\le \int_z^y \left(\frac{\kap(r+\rho)}{2}\frac{x}{H\big(J(-x)\big)^2(1+\rho x)^2}
	+ \frac{\kap}{H\big(J(-x)\big)(1+\rho x)}\right)\dd\eta\\
	&\le \frac{\kap}{H_0}\int_z^y \left(\frac{a_1 x}{(1+\rho x)^2}
	+ \frac{1}{1+\rho x}\right)\dd\eta\\
	\label{eq:ab_inequality}
	&= \frac{\kap}{\rho^2H_0}\left[\frac{a_1}{1+\rho y}-\frac{a_1}{1+\rho z}+(a_1+\rho)\ln\left(\frac{1+\rho y}{1+\rho z}\right)\right]
	\le \frac{\kap(a_1+\rho)}{\rho^2H_0} \ln\left(\frac{1+\rho y}{1+\rho z}\right),
\end{align}\end{linenomath*}
in which $a_1:=\frac{r+\rho}{2H_0}$. Let $b_1:= \frac{2\kap(a_1+\rho)}{\rho^2H_0}$, and note that, because $0<H_0<\kap/\rho$, we have
\begin{linenomath*}\begin{align}
	b_1 = \frac{2\kap}{\rho^2H_0}\left(\frac{r+\rho}{2H_0}+\rho\right)
	\ge \frac{2\kap}{\rho^2\frac{\kap}{\rho}}\left(\frac{r+\rho}{2\frac{\kap}{\rho}}+\rho\right)
	= \frac{r+\rho}{\kap}+2 > 2.
\end{align}\end{linenomath*}
For $x>\xa$, \eqref{eq:H_inequality} and \eqref{eq:ab_inequality} yield
\begin{linenomath*}\begin{align}
	\psi(x) &= \int_{\xa}^x\int_{\xa}^y \frac{2}{a(z)^2}\exp\left(-2\int_z^y \frac{b(\eta)}{a(\eta)^2}\dd\eta\right) \dd z\dd y\\
	&\ge \int_{\xa}^x\int_{\xa}^y \frac{2}{\frac{2\kap}{\rho^2}(1+\rho z)^2}\exp\left(
	-\frac{2\kap(a_1+\rho)}{\rho^2H_0} \ln\left(\frac{1+\rho y}{1+\rho z}\right)\right) \dd z\dd y\\
	&= \int_{\xa}^x\int_{\xa}^y \frac{\rho^2}{\kap}(1+\rho z)^{-2}
	\left(\frac{1+\rho z}{1+\rho y}\right)^{b_1} \dd z\dd y
	= \frac{\rho}{\kap(b_1-1)}\int_{\xa}^x\left(\frac{1}{1+\rho y}- \frac{(1+\rho\xa)^{b_1-1}}{(1+\rho y)^b_1}\right) \dd y\\
	&= \frac{\rho}{\kap(b_1-1)}\left[\frac{1}{\rho}\ln\left(\frac{1+\rho x}{1+\rho\xa}\right)+\frac{\rho(1+\rho\xa)^{b_1-1}}{b_1-1}\left((1+\rho x)^{1-b_1}-(1+\rho\xa)^{1-b_1}\right)\right].
\end{align}\end{linenomath*}
Finally, by letting $x\to+\infty$, it follows that $\lim_{x\to+\infty} \psi(x) = +\infty$.
\end{proof}

We end this section by proving certain properties of the optimal policy.

\begin{corollary}\label{cor:thts_c_props}
	The optimal relative consumption policy $\cs(x)$ is increasing in $x$. The optimal relative investment policy $\thts(x)$ is asymptotically linear in $x$. Specifically, $\lim_{x\to+\infty}\thts(x)/x = \beta\frac{\mu-r}{\sig^2}$ for some constant $\beta\in(0,1]$.
\end{corollary}
\begin{proof}
	That $\cs(x)$ is increasing follows from \eqref{eq:CS-sol}, the fact that $\varphi(y)$ is increasing by Proposition \ref{prop:FBP-sol}.(i), and that $J(\xi)$ is increasing since its inverse $u'(y)$ is increasing by Proposition \ref{prop:FBP-sol}.(ii). To obtain the second statement, note that $\beta:= \frac{\rho}{\kap}\lim_{y\to0^+}H(y)\in(0, 1]$ by Proposition \ref{prop:FBP-sol}.(i). From \eqref{eq:thts}, we then obtain that
	\begin{align}
		\lim_{x\to+\infty} \frac{\theta^*(x)}{x} =
		\frac{\mu-r}{\sig^2}\lim_{x\to+\infty} H\big(J(-x)\big)\frac{1+\rho x}{\kap x}
		= \frac{\mu-r}{\sig^2}\beta.
	\end{align}
\end{proof}

%
%

\section{Numerical illustrations}\label{sec:numerics}

%
%

\begin{figure}[p]
\centerline{
	\adjustbox{trim={0.0\width} {0.0\height} {0.0\width} {0.0\height},clip}
	{\includegraphics[scale=0.36, page=1]{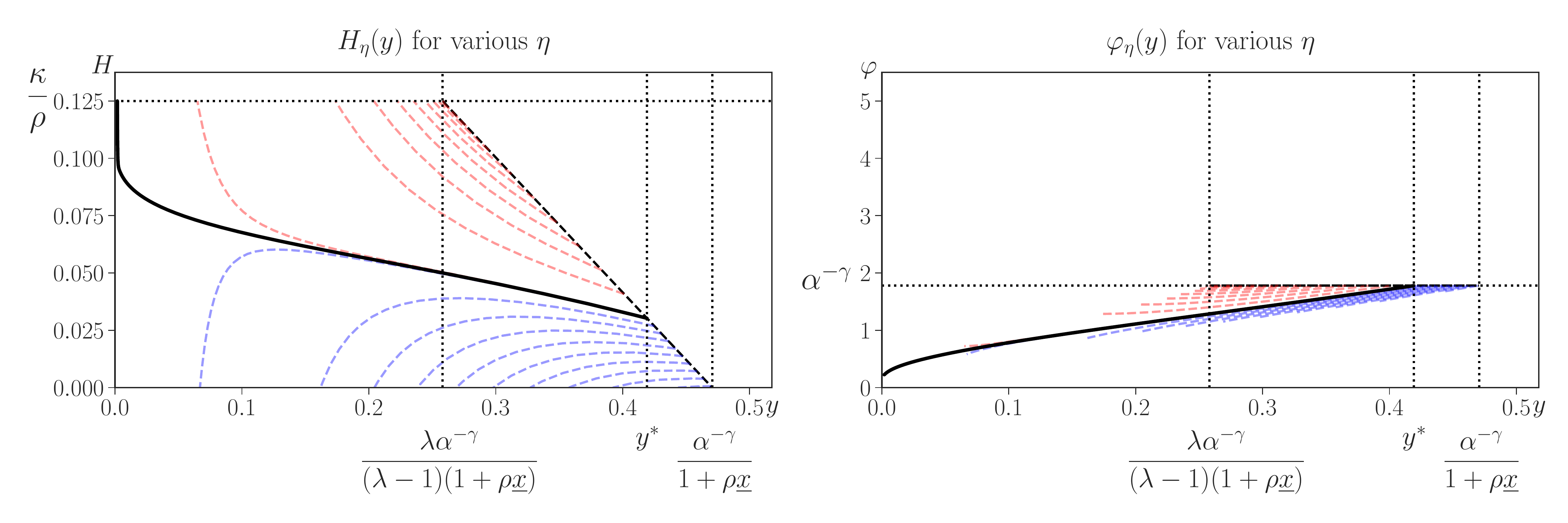}}
}
\caption{The solid black curves represent the (approximate) solution of the free-boundary problem \eqref{eq:phi-H-FBP-sys}. The dashed red and blue curves are the upper and lower solutions that satisfy the boundary-value problem \eqref{eq:phiH-eta} within the set $\Dc$ given by \eqref{eq:DC}. $\ya$ is the value of $\eta$ such that the solution exists for all $y\in(0,\eta)$.
	\label{fig:phi_H}}
\end{figure}
%
%
\begin{figure}[p]
\centerline{
	\adjustbox{trim={0.0\width} {0.02\height} {0.0\width} {0.02\height},clip}
	{\includegraphics[scale=0.36, page=1]{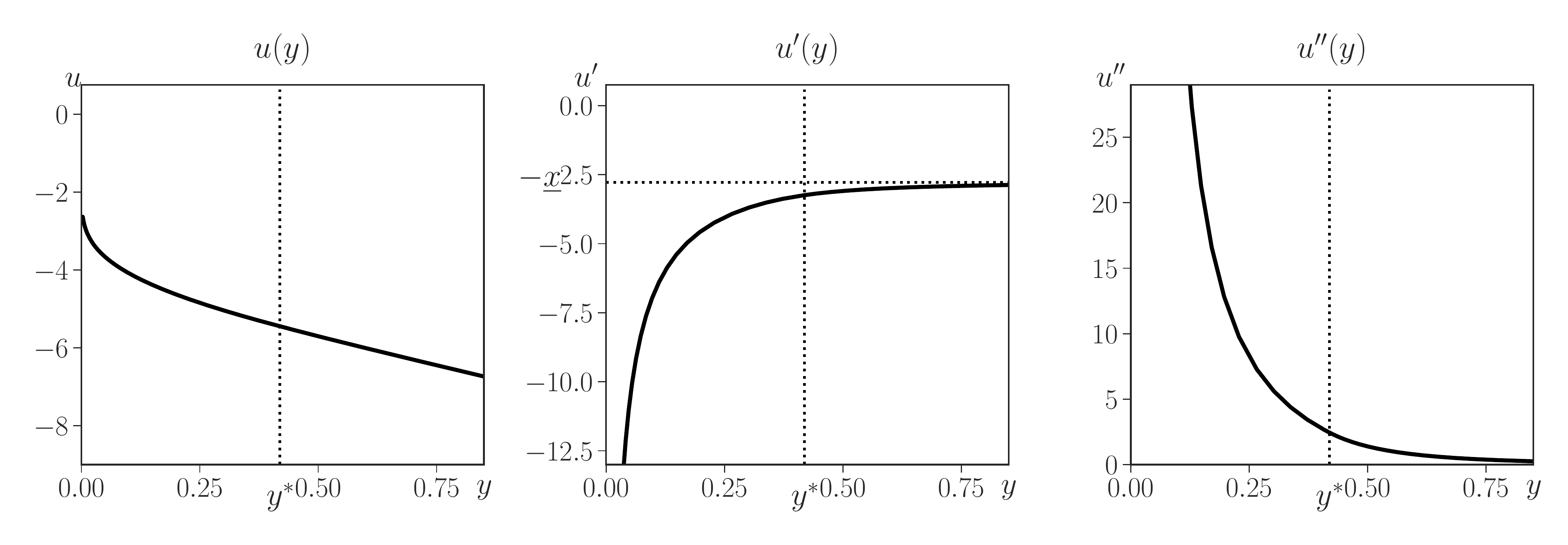}}
}
\caption{The solution $(\ya, u(\cdot))$ of FBP \eqref{eq:FBPDual1}-\eqref{eq:FBPDual5} and its first two derivatives.
	\label{fig:u}}
\end{figure}
%
%
\begin{figure}[p]
\centerline{
	\adjustbox{trim={0.0\width} {0.02\height} {0.0\width} {0.0\height},clip}
	{\includegraphics[scale=0.36, page=1]{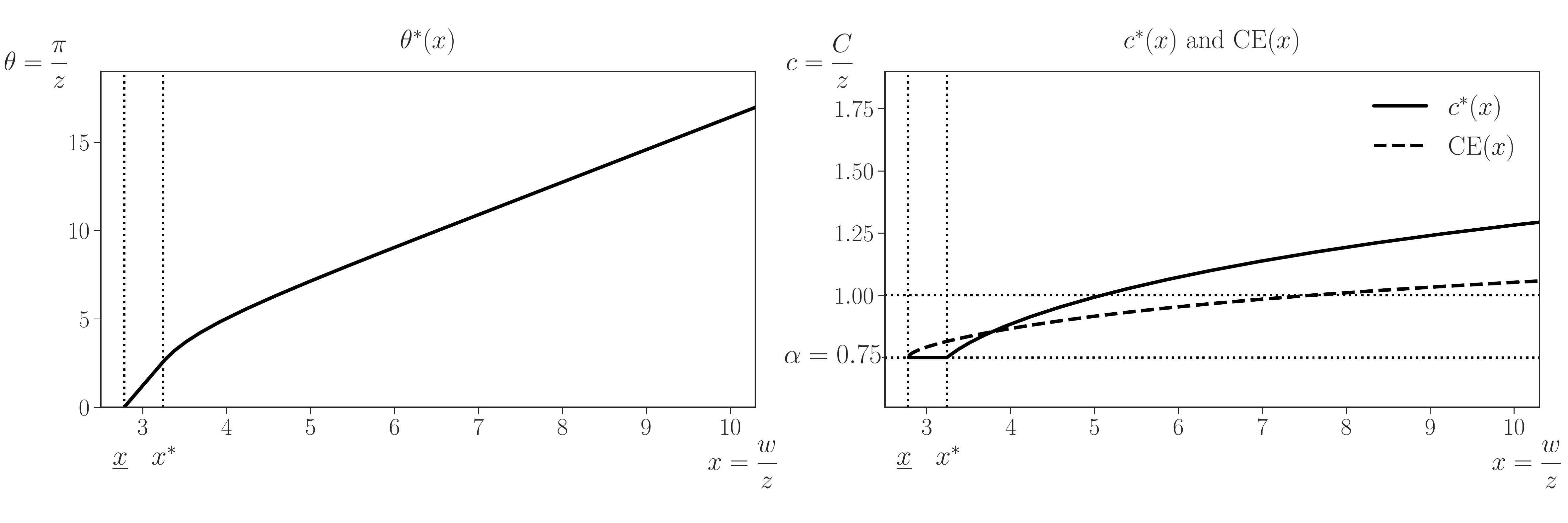}}
}
\caption{The optimal investment function $\thts(x)$, the optimal consumption function $\cs(x)$, the certainty equivalent function CE$(x)=\big(\del(1-\gam)V(x)\big)^{\frac{1}{1-\gam}}$.
	\label{fig:c_tht_Star}}
\end{figure}

We end the paper by providing a series of numerical examples to highlight certain properties of the optimal investment and consumption policy. Throughout the section, we choose the following values for the model parameters: $r=0.02, \mu=0.12, \sig=0.2, \rho=1, \al=0.75, \del=0.3,$ and $\gam=2$. On occasions, however, we will change the value of a parameter (while keeping other parameters fixed) to show sensitivity of the solution with respect to that parameter.

To obtain the solution, we first numerically solve FBP \eqref{eq:phi-H-FBP-sys} as follows. For a given value of $\ya$, \eqref{eq:phi-H-FBP-sys} can be solved using an ODE solver
(we used ``\texttt{RK45}'' through Python's \texttt{scipy.integrate.solve\_ivp()} function). By using a simple bisection search, we then find the smallest value of $\ya\in(\eta_1, \eta_2)$\footnote{Recall that $\eta_i$ are the constants in \eqref{eq:eta1eta2}.} for which $H$ exits from the top boundary $H=\kap/\rho$. The algorithm is illustrated by Figure \ref{fig:phi_H}. With $\ya$, $H$, and $\varphi$ at hand, we can use \eqref{eq:u_ge_ys} to find $u(y)$ and its first two derivatives for all $y>0$, as shown in Figure \ref{fig:u}.

Proposition \eqref{prop:HJB} then yields $\xa$, $v$, $\cs$, and $\thts$. The left plot of Figure \ref{fig:c_tht_Star} shows the optimal investment function $\thts(x)$. As indicated by \eqref{eq:thts}, for $x\in[\xu, \xa]$, $\thts(x)$ is linear with slope $\frac{\mu-r}{\sig^2}(1-\lam)$. For $x>\xa$, $\thts(x)$ is asymptotically linear with the slope $\frac{\mu-r}{\sig^2}$ since $\lim_{x\to+\infty} H\big(J(-x)\big)=\kap/\rho$. Indeed, as Figure \ref{fig:c_tht_Star} shows, this asymptotic linearity can occur for small values of $x$. Since $\lam<0$, the slope of $\thts(x)$ is greater in the range $x\in[\xu, \xa]$ than in the range $x>\xa$. In other words, the individual invests extra wealth more aggressively when her wealth-to-habit ratio is below the critical level $\xa$ compared to when her relative wealth is above $\xa$.

The right plot in Figure \ref{fig:c_tht_Star} shows the optimal consumption function $\cs(x)$ by the solid black curve. As indicated by \eqref{eq:CS-sol}, the optimal policy is to consume at the lowest consumption to habit ratio of $\al$ while wealth-to-habit ratio is below $\xa$ and to increase relative consumption once the relative wealth becomes larger than $\xa$. In the same plot, the dashed curve represents the \emph{certainty equivalent} (CE) function, which we define as follows. Assume that the individual maintains a constant consumption-to-habit ratio of $\ct$. Then, her utility of this consumption stream is
\begin{linenomath*}\begin{align}
\int_{0}^{+\infty} \ee^{-\del t} \frac{\ct^{1-\gam}}{1-\gam} \dd t = \frac{\ct^{1-\gam}}{\del(1-\gam)}.
\end{align}\end{linenomath*}
We define $CE(x)$ as the value of the constant consumption-to-habit process that yields the same utility as $V(x)$ of \eqref{eq:VF}. In other words, the individual is indifferent between receiving a constant consumption-to-habit ratio of CE$(x)$ versus consuming according to Theorem \ref{thm:VF}. It follows that we must have
\begin{linenomath*}\begin{align}
\frac{\text{CE}(x)^{1-\gam}}{\del(1-\gam)} = V(x)
\quad\Longrightarrow\quad
\text{CE}(x) = \big(\del(1-\gam)V(x)\big)^{\frac{1}{1-\gam}}.
\end{align}\end{linenomath*}
From the plot, we observe that the optimal consumption and CE functions meet at a point $(x_0, c_0)\approx(3.8, 0.85)$ such that $\cs(x)<\text{CE}(x)$ (resp.\ $\cs(x)>\text{CE}(x)$) for $x\in(\xu,x_0)$ (resp.\ $x>x_0$). Thus, by following the optimal consumption policy, the individual consumes less than (resp.\ greater than) her ``overall'' consumption rate if her wealth-to-habit ratio is below (resp.\ above) the relative wealth $x_0$. This observation indicates that the individual has a preference for specific levels of consumption-to-habit and wealth-to-habit ratios. In \cite{AngoshtariBayraktarYoung2020}, for the case when risky investment is not allowed, we showed a strong form of this property and explicitly identified the corresponding relative wealth and consumption levels $(x_0, c_0)$.

In Figure \ref{fig:xs_del_sens}, we investigate the dependence of the critical wealth-to-habit ratio $\xa$ on the subjective discount rate $\del$ in \eqref{eq:VF}.  We find $\xa$ to be decreasing in $\del$, which indicates that impatient individuals (that is, with higher $\del$) are more eager to consume at a rate higher than $\al$ than patient individuals (that is, with lower $\del$). We also saw this relationship in \cite{AngoshtariBayraktarYoung2020} for the case of riskless investment. In \cite{AngoshtariBayraktarYoung2020}, we also found that $\xa=\xu$ for $\del\ge r+\rho(1-\al)$. In contrast, Figure \ref{fig:xs_del_sens} highlights that $\xa>\xu$ for all values of $\del>0$, which we proved in Section \ref{sec:optimal}. Indeed, Proposition \ref{prop:FBP-sol}.(i) implies that $\ya<\eta_2$, from which it follows that $\xa>\xu$ by \eqref{eq:xa}. 

%
%

\begin{figure}[t]
\centerline{
	\adjustbox{trim={0.0\width} {0.02\height} {0.0\width} {0.0\height},clip}
	{\includegraphics[scale=0.35, page=1]{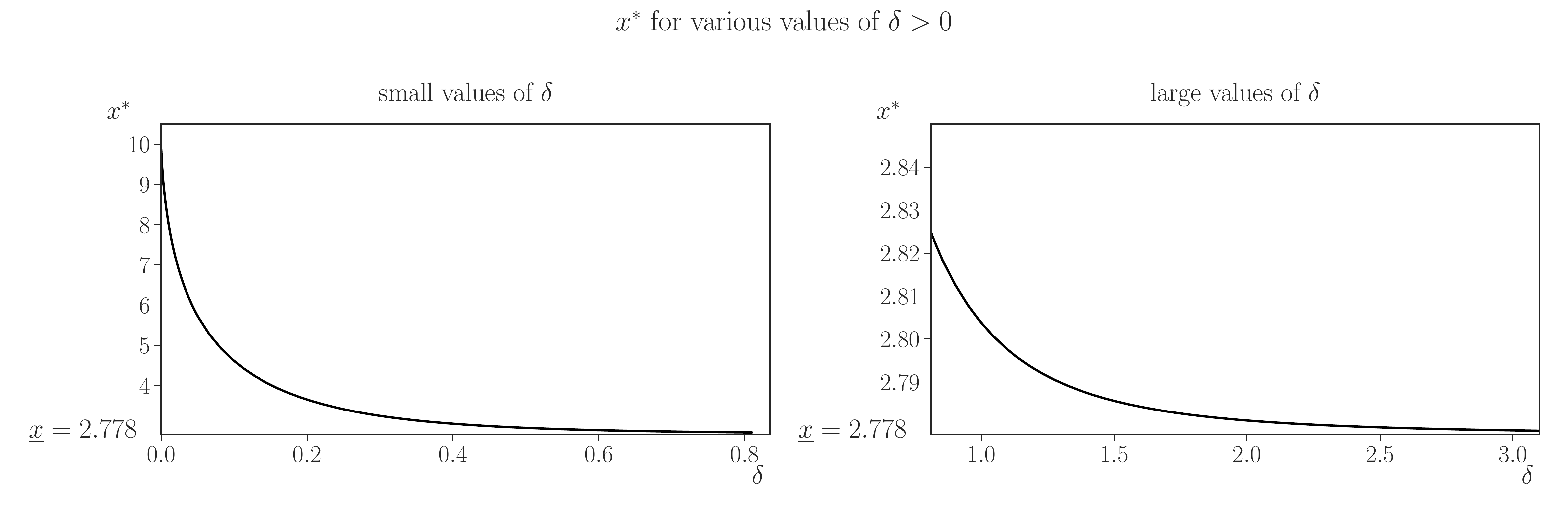}}
}
\caption{Sensitivity of the critical threshold $\xa$ with respect to $\delta$. Because of the difference in scale of $\xa$ values, we have separated the plot for small (on left) and large (on right) values of $\del$. Note that the lowest range of the vertical axes is $\xu$ and not zero. 
	\vspace{1em}
	\label{fig:xs_del_sens}}
\end{figure}

%
%

\begin{figure}[t]
	\centerline{
			\adjustbox{trim={0.0\width} {0.02\height} {0.0\width} {0.0\height},clip}
			{\includegraphics[scale=0.47, page=1]{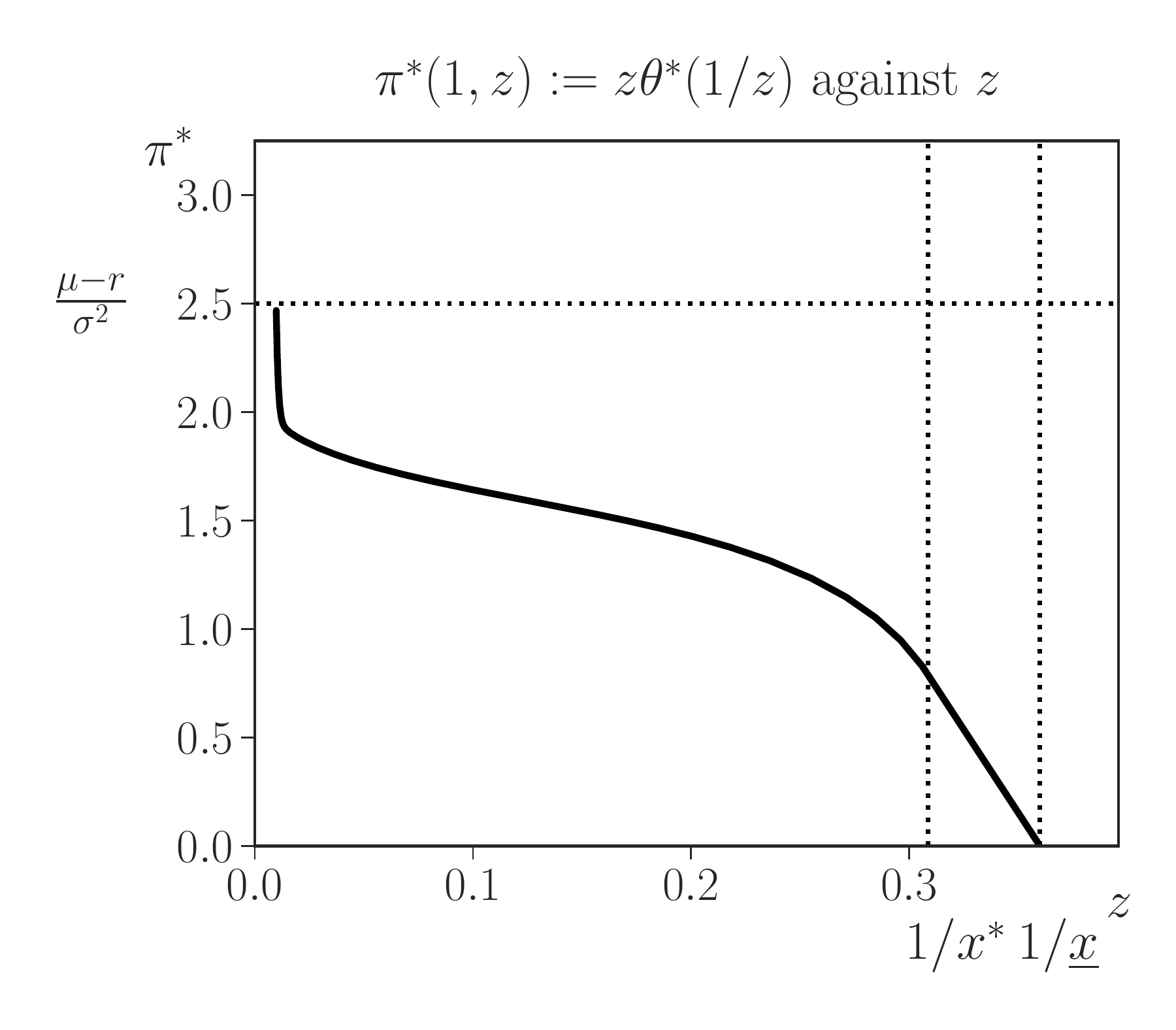}}
		\adjustbox{trim={0.0\width} {0.02\height} {0.0\width} {0.0\height},clip}
		{\includegraphics[scale=0.47, page=1]{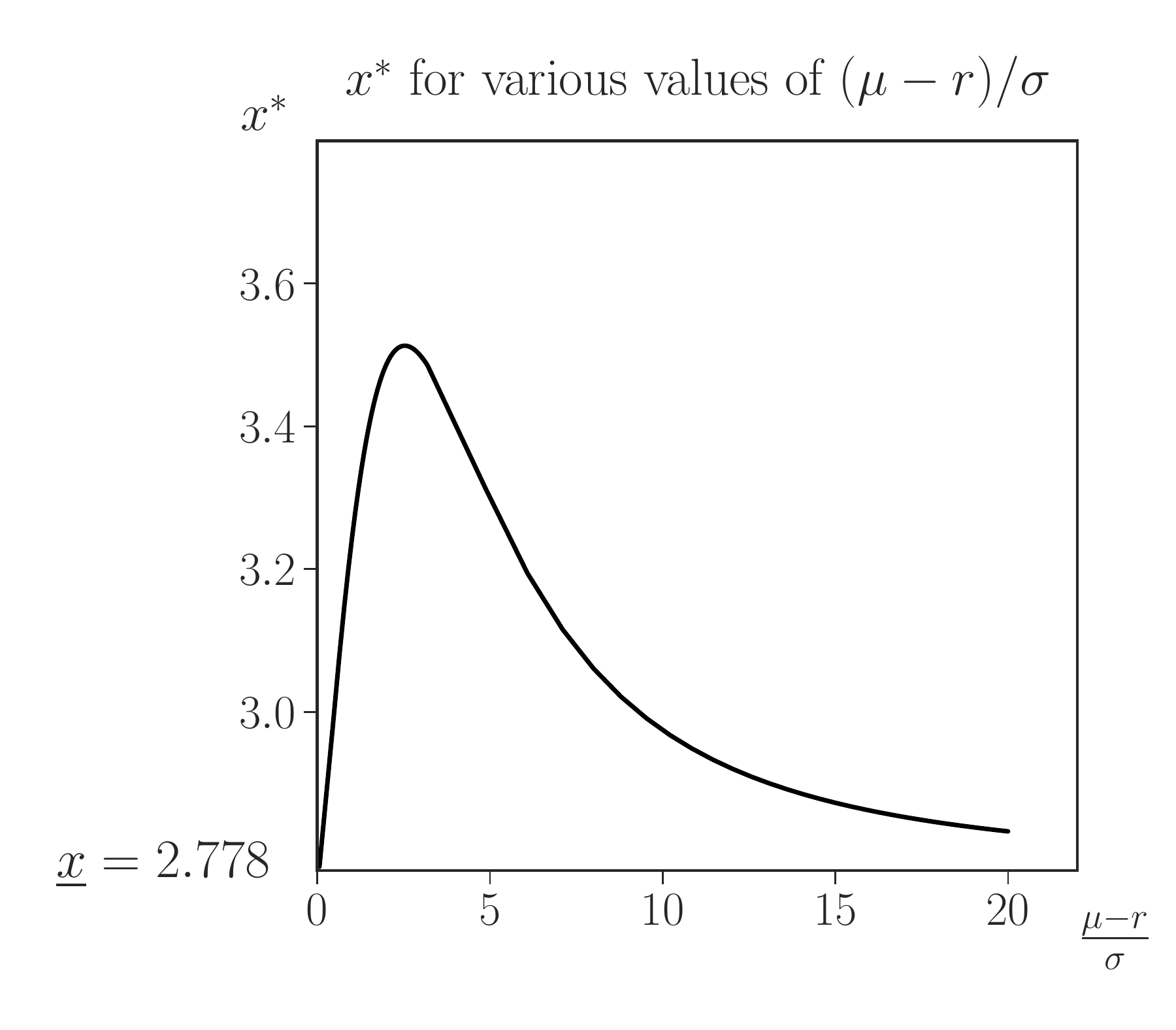}}
	}
	\caption{\textbf{Left:} Plot of the optimal (absolute) investment policy as a function of consumption habit and for fixed wealth, that is, the plot of $\pi^*(w,z) = z\thts(w/z)$ for $w=1$. Note that the plot is not defined in a neighborhood of $z=0$ because our numerical approximation for $H(y)$ is not defined for very small value of $y$ (i.e. very large values of $x$). See Figure \ref{fig:phi_H}.
		\textbf{Right:} Sensitivity of the critical threshold $\xa$ with respect to the ratio $(\mu-r)/\sig$. Note that the lowest range of the vertical axes is $\xu$ and not zero. 
		\label{fig:pi_Z}}
	\vspace{1em}
\end{figure}

The left plot in Figure \ref{fig:pi_Z} shows dependence of the optimal absolute investment policy $\pi^*_t$ on the consumption habit. Let $W^*_t$ and $Z^*_t$ be the optimally controlled wealth and consumption habit. By Proposition \ref{prop:Admiss}, the optimal investment in the stock is $\pi^*_t=Z^*_t\thts(W^*_t/Z^*_t)$. The left plot of Figure \ref{fig:pi_Z} shows the function $\pi^*(w,z):=z\thts(w/z)$ against the habit $z$ and for the fixed wealth $w=1$. Note that, for $w=1$, we must have $z\ge 1/\xu$ by \eqref{eq:NoRuin}. For $\frac{1}{\xa}\le z\le \frac{1}{\xu}$ (equivalently $x=w/z\in[\xu,\xa])$, \eqref{eq:thts} yields that $\pi^*(1,z)=\dfrac{(\mu-r)(1-\lam)}{\sig^2} \, (1-\xu\,z)$, so, $\pis$ is linear in $z$ for this range.
For $0<z\le 1/\xa$, the plot shows that $\pi^*$ increases as $z$ decreases, and it seems that $\pi^*$ has a limit in the interval $(0,(\mu-r)/\sig^2]$ as $z\to0^+$. The latter statement follows from Corollary \ref{cor:thts_c_props} as follows:
\begin{align}
	\lim_{z\to0^+} z\,\thts\left(\frac{1}{z}\right) = \lim_{x\to+\infty} \frac{\thts(x)}{x} = \frac{\mu-r}{\sig^2} \beta, 
\end{align}
for the constant $\beta\in(0,1]$ in the corollary. Note, also, that our numerical solution is not accurate as $z\to 0^+$ (equivalently, $x\to+\infty$), since our approximation of $H(y)$ is not accurate as $y\to0^+$.

The right plot of Figure \ref{fig:pi_Z} shows sensitivity of the threshold $\xa$ on the expected return $\mu$ and volatility $\sigma$ of the risky asset. By Propositions \ref{prop:FBP-sol}.(i) and \eqref{eq:xa}, $\mu$ and $\sigma$ affect $\xa$ through $\kap=\dfrac{(\mu-r)^2}{2\sig^2}$. Thus, it suffices to investigate the dependence of $\xa$ on the value of $\kap$ or, equivalently, on the Sharpe ratio (SR) $(\mu-r)/\sig=\sqrt{2\kap}$. The right plot of Figure \ref{fig:pi_Z} shows that $\xa$ is increasing for small values of SR, and it is decreasing for large values of SR. 

We interpret this result as follows. For small values of SR, the investor mostly uses the riskless asset for building up her wealth. Thus, her optimal consumption policy is close to the one studied by \cite{AngoshtariBayraktarYoung2020}, who showed that the threshold $\xa$ is close to $\xu$ (indeed, impatient individuals with $\del<r+\rho(1-\al)$ would have $\xa=\xu$). If SR increases, the investor would start using the risky asset and will be willing to wait longer before increasing her consumption above its minimum. Thus, $\xa$ is increasing in SR for small values of SR. If SR is sufficiently large, however, increasing SR would enable the investor to reach her ideal wealth-to-habit ratio more quickly, and thus, she could afford to consume above her minimum rate sooner. Thus, $\xa$ is decreasing in SR for large values of SR.

Figure \ref{fig:cs_thts_alpha_sens} shows dependence of the optimal policy on the parameter $\al$ in \eqref{eq:Habit}. Note that, by \eqref{eq:xu}, $\xu$ is increasing in $\al$. Thus, the domains of $\cs$ and $\thts$ in Figure \ref{fig:cs_thts_alpha_sens} shift to right as $\al$ increases. The top-left plot indicates that increasing $\al$ decreases the optimal investment-to-habit ratio $\thts(x)$, as long as the current level of wealth-to-habit ratio stays admissible (that is, $x\ge\xu$). The top-right plot shows that an increase in $\al$ increases (resp.\ decreases) $\cs(x)$ if $x\in(\xu, \xa)$ (resp.\ $x>\xa$). In other words, an individual who is more amenable to addiction (that is, higher $\al$) optimally invests less in the risky asset than an individual with less addictive personality and the same wealth-to-habit ratio. Furthermore, the individual with more addictive personality optimally consumes less than the individual with less addictive personality, unless the former individual's consumption is driven by the habit-formation constraint (that is, $x\in(\xu, \xa)$ such that $\cs(x)=\al$ for the individual with higher $\al$).

In the bottom plots of Figure \ref{fig:cs_thts_alpha_sens}, we investigate the asymptotic behavior of $\cs(x)$ and $\thts(x)$ for large values of $x$. The bottom-left plot is the log-log plot of $\thts(x)$ which shows that the optimal investment-to-habit ratio has linear growth in wealth-to-habit ratio $x$ (as indicated by Corollary \ref{cor:thts_c_props}). The bottom-right plot is the log-log plot of $\cs(x)$ which shows that the optimal consumption-to-habit ratio has sub-linear growth in wealth-to-habit ratio $x$. These plots also indicate that $\thts$ and $\cs$ are asymptotically independent of the value of $\al$ (as $x\to+\infty$), which is expected since the habit-formation constraint $C_t\ge \al Z_t$ (or, equivalently $X_t\ge \al$) should be asymptotically redundant for large $x$. Note, however, that removing the habit formation constraint will not yield the standard Merton problem because of dependence of our objective function \eqref{eq:EU} on the habit process $Z_t$. Indeed, the asymptotic model (as $\al\to0^+$) will be
\begin{align}\label{eq:VF_al0}
	V(x) = \sup_{\theta, c} \Eb_{x}\left(\int_0^{+\infty} \frac{c_t^{1-\gam}}{1-\gam} \,\ee^{-\del t}\, \dd t\right);\quad x>0,
\end{align}
with $(X_t,\theta_t,c_t)_{t\ge0}$ satisfying \eqref{eq:X-SDE}. To the best of our knowledge, the stochastic control problem \eqref{eq:VF_al0} has only been considered in Section 2.3 of \cite{2013Rogers} who only provided limited numerical results showing that the investment and consumption policies are very different from those in the classical Merton problem. As in our model, the numerical results in \cite{2013Rogers} indicate that $\theta^*(x)$ has linear growth (like the Merton problem) and that $c^*(x)$ has sublinear growth (unlike the Merton problem).

%
%

\begin{figure}[p]
	\centerline{
			\adjustbox{trim={0.0\width} {0.5\height} {0.0\width} {0.0\height},clip}
			{\includegraphics[scale=0.34, page=1]{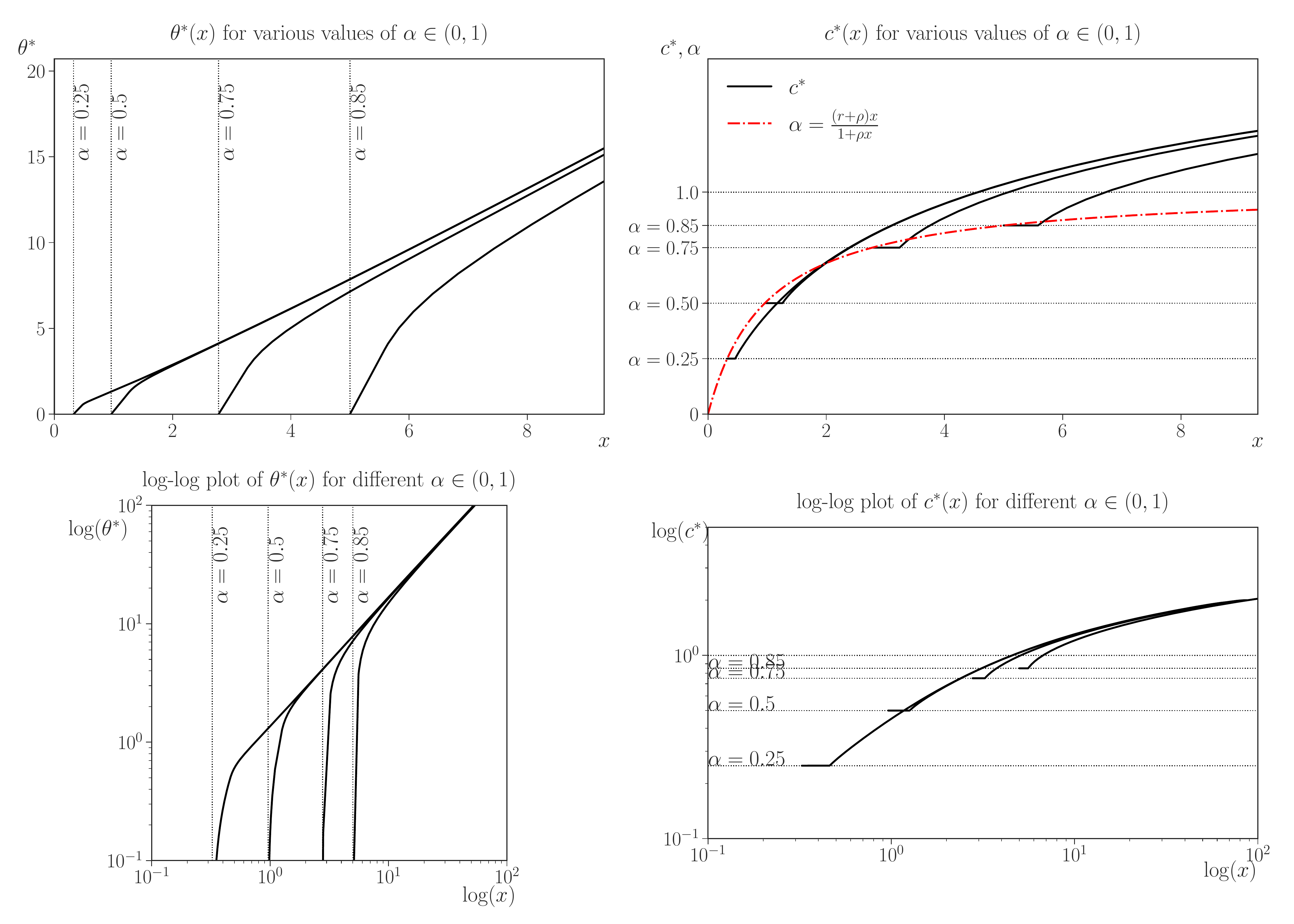}}
	}
	\vspace{1ex}
	\centerline{
			\adjustbox{trim={0.0\width} {0.02\height} {0.0\width} {0.5\height},clip}
			{\includegraphics[scale=0.34, page=1]{cs_thts_alpha_sens.pdf}}
	}
	\caption{\textbf{Top:} Sensitivity of the optimal investment function $\thts(x)$ and the optimal consumption function $\cs(x)$ with respect to $\alpha$. \textbf{Bottom:} Log-log plots corresponding to the top plots. For large values of $x$, $\theta^*(x)$ shows linear growth in $x$ while $c^*(x)$ shows sublinear growth. 
		\label{fig:cs_thts_alpha_sens}}
	\vspace{2em}
	\centerline{
			\adjustbox{trim={0.0\width} {0.02\height} {0.0\width} {0.0\height},clip}
			{\includegraphics[scale=0.34, page=1]{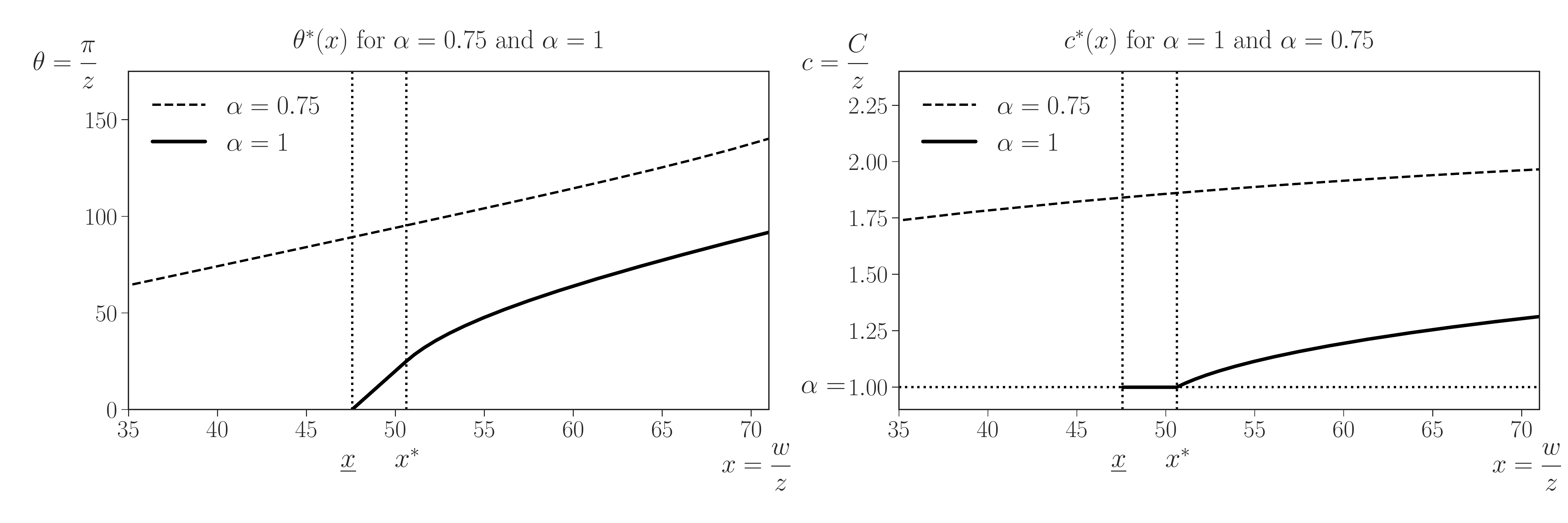}}
	}
	\caption{The optimal investment function $\thts(x)$ and the optimal consumption function $\cs(x)$ for the addictive habits, that is, $\al=1$. For reference, the optimal investment and consumption function for the nonaddictive habits $\al=0.75$ are also shown by the dashed curves. Note that the left limit on the horizontal axis is $x=35$.
		\label{fig:cStar_1_alpha}}
\end{figure}

Figure \ref{fig:cStar_1_alpha} shows the optimal policies for the case $\al=1$, which was included in the analysis of Section \ref{sec:optimal}. For this case, the individual's consumption rate is forced to be at least as large as her habit by \eqref{eq:Habit}, that is, $C_t\ge Z_t$. This scenario is usually referred to as \emph{addictive} habit formation, while the case in which $C_t<Z_t$ is allowed is called \emph{nonaddictive} habit formation.\footnote{See, for instance, \cite{DetempleKaratzas2003}.} Therefore, in our setting, $\al=1$ (resp.\ $\al<1$) represents addictive (resp.\ nonaddictive) habit formation. As Figure \ref{fig:cStar_1_alpha} shows, the optimal policies of the addictive and nonaddictive cases have a similar structure. Their main difference is that the amount of wealth needed to support a certain level of consumption is significantly higher for addictive habits. For instance, for our chosen parameter values, an addictive individual needs a wealth that is 47 times her habit to avoid bankruptcy (that is, $\xu\approx 47$ in the right plot of Figure \ref{fig:cStar_1_alpha}), and a wealth of about 50 times her habit to consume above the minimum rate. On the other hand, Figure \ref{fig:cs_thts_alpha_sens} shows that for a nonaddictive habit formation with $\al=0.75$, the individual needs a wealth-to-habit ratio of around 3 to optimally consume above her minimum rate.

Finally, Figure \ref{fig:cStar_1_alpha} shows that, for the same values of risk aversion and wealth-to-habit ratio, addictive habits (that is, $\al=1$) correspond to significantly lower levels of optimal consumption and optimal investment in the risky asset than nonaddictive habits (with $\al=0.75$). In other words, individuals with more addictive habits (optimally) invest less in the risky asset. To attract such individuals, the market premiums needs to be higher than they would be for individuals with less addictive habits. This observation provides an explanation for the \emph{equity premium puzzle} of \cite{mehra1985equity}, which states that the historical risk premium offered by stock markets has been significantly higher than the level that could be explained by investors' risk aversion alone. See \cite{Constantinides1990} for further discussion on the puzzle and how it can be explained by habit-formation models.

%
%

\bibliographystyle{chicago}
\bibliography{references}

%
%
\appendix

%
%
\section{Proof of Theorem \ref{thm:verif}}\label{app:verif}

We complete the proof in two steps by showing (1) $v \ge V$ and (2) $v \le V$.
\smallskip

\noindent \textbf{Step 1:}  Let $(\theta_t,c_t)_{t\ge0} \in \Ac(\al)$ and $\{X_t \}_{t \ge 0}$ be the corresponding wealth-to-habit process given by \eqref{eq:X-SDE}. Define the non-decreasing sequence of stopping times $\{\tau_n\}_{n=1}^\infty$ by
\begin{linenomath*}\begin{align}\label{eq:tau_n}
\tau_n := \inf \left\{ t \ge 0: \int_0^t \ee^{-\del s} \big(\theta_s v'(X_s)\big)^2 \dd s \ge n \right\},
\end{align}\end{linenomath*}
for $n\ge 1$. For all $T\ge0$, applying It\^o's lemma to $\ee^{-\del t} v(X_t)$, $t\in[0,T\wedge \tau_n]$ yields
\begin{linenomath*}\begin{align}
&\ee^{-\del (T\wedge\tau_n)} v\big(X_{T\wedge\tau_n}\big) + \int_0^{T\wedge\tau_n} \frac{c_t^{1-\gam}}{1-\gam} \, \ee^{-\del t} \, \dd t = v(x) + \int_0^{T\wedge\tau_n} \ee^{-\del t} \, \Lc_{\theta_t, c_t} v(X_t) \, \dd t
+ \int_0^{T\wedge\tau_n} \sig \theta_t \ee^{-\del t} v'(X_t) \, \dd B_t.
\end{align}\end{linenomath*}
Condition $(i)$ implies that the first integral on the right is non-positive; thus, we have
\begin{linenomath*}\begin{align}
\frac{\al^{1-\gam}}{\del(1-\gam)}
\le \ee^{-\del (T\wedge\tau_n)} v\big(X_{T\wedge\tau_n}\big) + \int_0^{T\wedge\tau_n} \frac{c_t^{1-\gam}}{1-\gam} \, \ee^{-\del t} \, \dd t
\le v(x) + \int_0^{T\wedge\tau_n} \sig \theta_t \ee^{-\del t} v'(X_t) \, \dd B_t,
\end{align}\end{linenomath*}
in which we used condition $(ii)$ to get the first inequality.
The definition of $\tau_n$ implies that the expectation of the remaining integral on the right is zero, which implies
\begin{linenomath*}\begin{align}\label{eq:finite_n}
\frac{\al^{1-\gam}}{\del(1-\gam)} \le \Eb_x \left(\ee^{-\del (T\wedge\tau_n)} v\left(X_{T\wedge\tau_n}\right) + \int_0^{T\wedge\tau_n} \frac{c_t^{1-\gam}}{1-\gam} \, \ee^{-\del t} \, \dd t \right)
\le v(x).
\end{align}\end{linenomath*}
Define $\tau_\infty:=\displaystyle\esup\{\tau_n:n\ge1\}$, in which we include the possibility of $\Pb\left(\tau_\infty=+\infty\right)>0$.  From the dominated convergence theorem, because $\{\tau_n\}$ is non-decreasing, we deduce
\begin{linenomath*}\begin{align}\label{eq:E_tau_n}
\lim_{n\to+\infty} \Eb_x\Big(\ee^{-\del (T\wedge\tau_n)}\Big) = \Eb_x\Big(\ee^{-\del (T\wedge\tau_\infty)}\Big)\in[0,1).
\end{align}\end{linenomath*}
Because $v'(\xu^+)=+\infty$ by condition $(ii)$, we have $\tau_\infty < +\infty$ only if $X_{\tau_\infty}=\xu$ which, in turn, is equivalent to $X_t=\xu$ and $c_t=\al$ for all $t\ge \tau_\infty$ by the proof of Lemma 2.2 in \cite{AngoshtariBayraktarYoung2020}. By letting $n\to\infty$ in \eqref{eq:finite_n} and by using the dominated convergence theorem to exchange expectation and limit, we obtain
\begin{linenomath*}\begin{align}
&\frac{\al^{1-\gam}}{\del(1-\gam)} 
\le\lim_{n\to+\infty}\Eb_x \left(\ee^{-\del (T\wedge\tau_n)} v\left(X_{T\wedge\tau_n}\right) + \int_0^{T\wedge\tau_n} \frac{c_t^{1-\gam}}{1-\gam} \, \ee^{-\del t} \, \dd t \right)\\*
&= \Eb_x\left[
\mathds{1}_{\{ \tau_\infty < T \}}
\left(\ee^{-\del \tau_\infty} \frac{\al^{1-\gam}}{\del(1-\gam)} + \int_0^{\tau_\infty} \frac{c_t^{1-\gam}}{1-\gam} \, \ee^{-\del t} \, \dd t \right)
+ \mathds{1}_{\{ \tau_\infty \ge T \}}
\left(\ee^{-\del T} v(X_{T}) + \int_0^{T} \frac{c_t^{1-\gam}}{1-\gam} \, \ee^{-\del t} \, \dd t \right)
\right]\\
&= \Eb_x\left[
\Ib_{\{ \tau_\infty < T \}}\int_0^{+\infty} \frac{c_t^{1-\gam}}{1-\gam} \, \ee^{-\del t} \, \dd t
+\Ib_{\{ \tau_\infty \ge T \}}\int_0^{T} \frac{c_t^{1-\gam}}{1-\gam} \, \ee^{-\del t} \, \dd t
\right]\\
\label{eq:finite_T}
&\quad{}+\Eb_x\left[\mathds{1}_{\{ \tau_\infty \ge T \}}\ee^{-\del T} v(X_{T})\right]
\le v(x).
\end{align}\end{linenomath*}
To get the first equality, we used $X_{T\wedge\tau_\infty}=\xu$ when $\tau_\infty<T$ and $\vp(\xu) = \frac{\al^{1-\gam}}{\del(1-\gam)}$ from condition $(ii)$.  The second equality holds since, if $\tau_\infty<T$, then we have $c_t=\al$ for all $t \ge \tau_\infty$ and, thus, 
$\ee^{-\del \tau_\infty} \frac{\al^{1-\gam}}{\del(1-\gam)} = \int_{\tau_\infty}^{+\infty} \frac{c_t^{1-\gam}}{1-\gam} \, \ee^{-\del t} \, \dd t$. Next, we use condition $(iii)$ to deduce that
\begin{linenomath*}\begin{align}
0\le \lim_{T\to +\infty} \Eb_x\left[\mathds{1}_{\{ \tau_\infty \ge T \}}\ee^{-\del T} v(X_{T})\right]
\le \lim_{T\to +\infty} \Eb_x\left[\ee^{-\del T} v(X_{T})\right]=0.
\end{align}\end{linenomath*} 
Thus, by taking the limit as $T\to+\infty$ in \eqref{eq:finite_T} and by using the dominated convergence theorem, it follows that
\begin{linenomath*}\begin{align}
\Eb_x\left[\int_0^{+\infty} \frac{c_t^{1-\gam}}{1-\gam} \, \ee^{-\del t} \, \dd t\right]
\le v(x).
\end{align}\end{linenomath*}
Finally, by taking the supremum over admissible policies, we deduce $v \ge V$ on $[\xu, +\infty)$.

\bigskip	

\noindent \textbf{Step 2:} 	For this step, consider the admissible policy $\big(\thts(X^*_t), \cs(X^*_t)\big)_{t\ge0}$, and define the stopping time $\widehat{\tau}_n$ by
\[
\widehat{\tau}_n := \inf \left\{ t \ge 0: \int_0^t \ee^{-\del s} \big(\theta^*_s v_x(X^*_s) \big)^2 ds \ge n \right\}.
\]
Then, by repeating the argument in Step 1 and by using condition $(iv)$, we obtain
\begin{linenomath*}\begin{align}
v(x) = \Eb_x \Bigg(\ee^{-\del (T\wedge\widehat{\tau}_n)} v\left(X_{T\wedge\widehat{\tau}_n}\right) + 
\int_0^{T\wedge\widehat{\tau}} \frac{c_t^{1-\gam}}{1-\gam} \, \ee^{-\del t} \, \dd t \Bigg)
\ge
\frac{\al^{1-\gam}}{\del(1-\gam)} \Eb_x\Big(\ee^{-\del (T\wedge\widehat{\tau}_n)}\Big).
\end{align}\end{linenomath*}
By arguing as in Step 1, and by taking the limit as $n \to + \infty$ and, then, as $T\to+\infty$, we have
\[
v(x) = \Eb_x \Bigg(\int_0^{+\infty} \frac{(\cs(X^*_t))^{1-\gam}}{1-\gam} \, \ee^{-\del t} \, \dd t \Bigg).
\]
Thus, because $v$ is the value function corresponding to an admissible policy, we deduce $v \le V$ on $[\xu, +\infty)$.

%
%
\section{Auxiliary lemmas for Section \ref{sec:optimal}}\label{app:aux_constrained}

The following Lemma is used in the proof of Proposition \ref{prop:FBP-sol}.

\begin{lemma}\label{lem:phiH-eta}
For $\eta\in(\eta_1,\eta_2)$, let $(\varphi_\eta(\cdot), H_\eta(\cdot)\big)$ be the solution of the boundary-value problem \eqref{eq:phiH-eta} such that $(\eps(\eta),\eta]$ is the maximal domain over which the solution exists within $\Dc$ given by \eqref{eq:DC}. We, then, have:
\begin{itemize}
	\item[$(i)$] If $\eps(\eta)>0$, then $(\varphi_\eta(\cdot), H_\eta(\cdot)\big)$ exits $\Dc$ either through the boundary $\overline{\Dc}_1$ given by  \eqref{eq:Dc1} or through the boundary $\overline{\Dc}_2$ given by
	\begin{linenomath*}\begin{align}\label{eq:Dc2}
		\overline{\Dc}_2 := \big\{(y,\varphi,0):y\in(0,\eta_2), \, \varphi\in(0,\al^{-\gam}) \big\}.
	\end{align}\end{linenomath*}
	
	\item[$(ii)$] For values of $\eta\in(\eta_1, \eta_2)$ that are sufficiently close to $\eta_1$, the solution $(\varphi_\eta(\cdot), H_\eta(\cdot)\big)$ exits $\Dc$ through $\overline{\Dc}_1$.
	
	\item[$(iii)$] For values of $\eta\in(\eta_1, \eta_2)$ that are sufficiently close to $\eta_2$, the solution $(\varphi_\eta(\cdot), H_\eta(\cdot)\big)$ exits $\Dc$ through $\overline{\Dc}_2$.
	
	\item[$(iv)$] Assume that $\eta,\eta'\in(\eta_1,\eta_2)$ are such that $\eta<\eta'$ and the solutions $(\varphi_\eta(\cdot), H_\eta(\cdot)\big)$ and $(\varphi_{\eta'}(\cdot), H_{\eta'}(\cdot)\big)$ do not have disjoint domains, that is $\max\big\{\eps(\eta'),\eps(\eta)\big\} < \eta$. Then, $\varphi_\eta(y)>\varphi_{\eta'}(y)$ and $H_\eta(y)>H_{\eta'}(y)$ for all $y\in\big(\max\big\{\eps(\eta'),\eps(\eta)\big\}, \, \eta]$. 
\end{itemize}
\end{lemma}
\begin{proof}
\noindent \underline{Proof of $(i)$}:
From the differential equation for $\varphi$ in \eqref{eq:phiH-eta}, we deduce that $\varphi_\eta'(y)>0$ for $y\in(\eps(\eta),\eta)$, since $H_\eta(y)<\kap/\rho$. So, it can only be possible for $(\varphi_\eta(\cdot), H_\eta(\cdot)\big)$ to exit $\Dc$ from the boundary
\begin{linenomath*}\begin{align}
	\overline{\Dc}_0 := \big\{(y,0,H):y\in(0,\eta_2), \, H\in(0,\kap/\rho) \big\},
\end{align}\end{linenomath*}
the boundary
\begin{linenomath*}\begin{align}
	\overline{\Dc}_1' := \big\{(y,\varphi,\kap/\rho):y\in(0,\eta_2), \, \varphi\in(0,\al^{-\gam}) \big\}.
\end{align}\end{linenomath*}
or the boundary $\overline{\Dc}_2$.
We can eliminate the possibility of exiting through the boundary $\overline{\Dc}_0$ by the following argument. On the contrary, suppose $(\varphi_\eta(\cdot), H_\eta(\cdot)\big)$ exits $\Dc$ thorough $\overline{\Dc}_0$, that is, $0<H_\eta(y)<\kap/\rho$ for $\eps(\eta)<y\le\eta$ and $\lim_{y\to\eps(\eta)^+}\varphi_\eta(y)=0$. For $y>0$, define $u_1(y)=(1+\rho\xu)y$ and $u_2(y)=0$. Note that $u_1\big(\eps(\eta)) > 0 = \lim_{y\to\eps(\eta)^+}\varphi_\eta(y)$ and $u_2\big(\eps(\eta)) = 0 < \lim_{y\to\eps(\eta)^+}\varphi_\eta(y)$. Furthermore, for $\eps(\eta)<y\le \eta$, we have 
\begin{linenomath*}\begin{align}
	u_1'(y) - g_1\big(y,u_1(y),u_2(y)\big) =
	0 = \varphi_\eta'(y) - g_1\big(y,\varphi_\eta(y),H_\eta(y)\big),
\end{align}\end{linenomath*}
and
\begin{linenomath*}\begin{align}\label{eq:u_2p}
	u_2'(y) - g_2\big(y,u_1(y),u_2(y)\big)
	= 0 - \frac{1}{y}\left((1+\rho\xu)^{-\frac{1}{\gam}} y^{-\frac{1}{\gam}} - \al \right)
	< 0 
	= \varphi_\eta'(y) - g_2\big(y,\varphi_\eta(y),H_\eta(y)\big),
\end{align}\end{linenomath*}
in which $g_1$ and $g_2$ are given by \eqref{eq:g1} and \eqref{eq:g2}, respectively. To get the first equality in \eqref{eq:u_2p}, we used $\frac{\rho\xu}{1+\rho\xu} = \al$ which follows from \eqref{eq:xu}. To get the inequality in \eqref{eq:u_2p}, we used $0<y\le \eta < \eta_2=\frac{\al^{-\gam}}{1+\rho\xu}$.  Because $g_1(y,\varphi, H)$ is decreasing in $H$ and $g_2(y,\varphi,H)$ is decreasing in $\varphi$, we can apply Lemma \ref{lem:sys-comparison}.$(i)$ below to conclude that $\varphi_\eta(\eta)\le (1+\rho\xu)\eta$. The last statement, however, contradicts the boundary condition in \eqref{eq:phiH-eta}, namely, $\varphi_\eta(\eta) = \al^{-\gam}$ and $\eta<\eta_1\Rightarrow (1+\rho\xu)\eta<\al^{-\gam}$. Thus, $(\varphi_\eta(\cdot), H_\eta(\cdot)\big)$ can only exit $\Dc$ through either $\overline{\Dc}_1'$ or $\overline{\Dc}_2$.

To finish proving $(i)$, it remains to show that $(\varphi_\eta(\cdot), H_\eta(\cdot)\big)$ cannot exit through the boundary
\begin{linenomath*}\begin{align}
	\overline{\Dc}'_1\backslash \overline{\Dc}_1 = \left\{\left(y,\varphi,\kap/\rho\right): y\in[\eta_1 ,\eta_2), \varphi\in(0, \al^{-\gam})\right\}.
\end{align}\end{linenomath*}
To show this statement, it suffices to show
\begin{linenomath*}\begin{align}\label{eq:H_w2}
	H_\eta(y)\le w_2(y);\quad \max\big\{\eps(\eta),\eta_1\big\}<y\le\eta,
\end{align}\end{linenomath*}
in which $w_2$ is defined by
\begin{linenomath*}\begin{align}
	w_2(y)= \frac{\kap}{\rho}\left[1-\lam\left(1-\frac{y}{\eta_1}\right)\right];\quad y\in(0,\eta_2).
\end{align}\end{linenomath*}
Recall that $\lam < 0$, and note that $\frac{\kap}{\rho} = w_2(\eta_1)>w_2(y)>w_2(\eta_2)=0$ for $y \in (\eta_1, \eta_2)$. To show inequality \eqref{eq:H_w2}, let $w_1(y)=\al^{-\gam}$ for $y \in (0, \eta_2)$.  From \eqref{eq:phiH-eta}, we have $\varphi_\eta(\eta)=w_1(\eta)$ and $H_\eta(\eta)=w_2(\eta)$. Furthermore, for $y\in\big(\max\big\{\eps(\eta), \eta_1\big\}, \eta\big]$, we have
\begin{linenomath*}\begin{align}
	w_1'(y) - g_1\big(y,w_1(y),w_2(y)\big) = 0 - \frac{\rho}{\kap y}\left(\frac{\kap}{\rho} - w_2(y)\right)\al^{-\gam}
	<0 = \varphi_\eta'(y) - g_1\big(y,\varphi_\eta(y),H_\eta(y)\big), 
\end{align}\end{linenomath*}
and
\begin{linenomath*}\begin{align}
	&w_2'(y) - g_2\big(y,w_1(y),w_2(y)\big)\\
	&= \frac{\kap\lam}{\rho\eta_1}
	- \frac{\rho}{\kap y}\left(\frac{\kap}{\rho}-w_2(y)\right)\left(\frac{\del-r-\rho(1-\al)}{\rho}-w_2(y)\right)
	-\frac{r+\rho}{\rho\al^{-\gam}} + \frac{\del}{\rho y}\\
	&= \frac{\kap\lam}{\rho\eta_1} -\frac{r+\rho}{\rho\al^{-\gam}} + \frac{\del}{\rho y}
	+ \frac{1}{y} \left(1 - \dfrac{y}{\eta_1} \right) \left(- \, \dfrac{\kap \lam^2}{\rho} + \frac{(\kap + r + \rho(1-\al) - \del)\lam}{\rho} + \dfrac{\kap \lam^2 y}{\rho \eta_1} \right)\\
	&= \frac{\kap\lam}{\rho\eta_1} -\frac{r+\rho}{\rho\al^{-\gam}} + \frac{\del}{\rho y}
	+ \frac{1}{y} \left(1 - \dfrac{y}{\eta_1} \right) \left(- \, \dfrac{\del}{\rho} + \dfrac{\kap \lam^2 y}{\rho \eta_1} \right)\\
	&= \frac{\kap\lam}{\rho\eta_1} -\frac{r+\rho}{\rho\al^{-\gam}} + \frac{\del}{\rho \eta_1} + \dfrac{\kap \lam^2}{\rho \eta_1} \left(1 - \dfrac{y}{\eta_1} \right) \\
	&= \dfrac{1}{\rho \eta_1} \left[ \kap \lam - \dfrac{\lam(r + \rho(1 - \al))}{\lam -1} + \del + \kap \lam^2  \left(1 - \dfrac{y}{\eta_1} \right) \right] \\
	&= \frac{\kap \lam^2}{\rho\eta_1}\left(1-\frac{y}{\eta_1}\right)\\
	&< 0 
	= H_\eta'(y) - g_2\big(y,\varphi_\eta(y),H_\eta(y)\big).
\end{align}\end{linenomath*}
In two steps of the calculation for $w_2'-g_2$, we used the fact that $\lam$ satisfies \eqref{eq:lam_quadratic_eq}, and we used the definition of $\eta_1$ in \eqref{eq:eta1eta2}. To get the last inequality, we used $y>\eta_1$. Finally, inequality \eqref{eq:H_w2} follows from Lemma \ref{lem:sys-comparison}.$(ii)$ below.	
\vspace{1em}

\noindent \underline{Proofs of $(ii)$ and $(iii)$}:
As $\eta\to\eta_1^+$, The boundary condition in \eqref{eq:phiH-eta} approaches the point $(y,\varphi,H)=(\eta_1, \al^{-\gam}, \kap/\rho)$, which lies on the boundary of $\overline{\Dc}_1$. Furthermore,
\begin{linenomath*}\begin{align}
	g_2\big(\eta_1, \al^{-\gam}, \kap/\rho)= \frac{r+\rho}{\rho\al^{-\gam}}-\frac{\del}{\rho \eta_1}
	= \frac{r+\rho}{\rho\al^{-\gam}}\left(1+\frac{\del(1-\lam)}{\lam(r+\rho(1-\al))}\right)
	=\frac{\kap(r+\rho)(\lam-1)}{\rho\al^{-\gam}\big(r+\rho(1-\al)\big)}<0,
\end{align}\end{linenomath*}
in which we used \eqref{eq:xu} and \eqref{eq:eta1eta2} to get the second equality, \eqref{eq:lam_quadratic_eq} to get the third equality, and $\lam<0$ to get the inequality. From continuous dependence of the solution $(\varphi_\eta(\cdot), H_\eta(\cdot)\big)$ on $\eta$, it follows that $(\varphi_\eta(\cdot), H_\eta(\cdot)\big)$ exits $\Dc$ through $\overline{\Dc}_1$ for values of $\eta$ in a right neighborhood $(\eta_1,\eta_1+\epsilon)$ of $\eta_1$. With a similar argument, we conclude that $(\varphi_\eta(\cdot), H_\eta(\cdot)\big)$ exits $\Dc$ through $\overline{\Dc}_2$ for values of $\eta$ in a left neighborhood $(\eta_2-\epsilon',\eta_2)$ of $\eta_2$.
\vspace{1em}

\noindent \underline{Proof of $(iv)$}: The statement directly follows from Lemma \ref{lem:sys-comparison}.$(ii)$ below by taking into account that $(\varphi_\eta(\cdot), H_\eta(\cdot)\big)$ and $(\varphi_{\eta'}(\cdot), H_{\eta'}(\cdot)\big)$ are unique solutions of \eqref{eq:phiH-eta}.
\end{proof}

We refer to the following lemma in the proof of Lemma \ref{lem:phiH-eta}.

\begin{lemma}\label{lem:sys-comparison}
For an open set $D\subseteq\Rb^2$ and an interval $J=(a,b)$, assume that the vector-valued function $(f_1,f_2)=\fv(x,\yv):J\times D\to \Rb^2$ is locally Lipschitz continuous with respect to $\yv$, that $f_1(x,y_1,y_2)$ is decreasing in $y_2$, and that $f_2(x,y_1,y_2)$ is decreasing in $y_1$. Let $\uv=(u_1,u_2):J\to D$ and $\wv=(w_1,w_2):J\to D$ be differentiable functions. Then:
\begin{itemize}
	\item[$(i)$] If $u_1(a^+)\ge w_1(a^+)$, $u_2(a^+)\le w_2(a^+)$, $u_1'(x)-f_1\big(x,\uv(x)\big)\ge w_1'(x)-f_1\big(x,\wv(x)\big)$, and $u_2'(x)-f_2\big(x,\uv(x)\big)\le w_2'(x)-f_2\big(x,\wv(x)\big)$ for $x\in J$, then $u_1(x)\ge w_1(x)$ and $u_2(x)\le w_2(x)$ for $x\in J$.
	
	\item[$(ii)$] If $u_i(b^-)\le w_i(b^-)$ and $u_i'(x)-f_i\big(x,\uv(x)\big)\ge w_i'(x)-f_i\big(x,\wv(x)\big)$ for $x\in J$ and $i\in\{1,2\}$, then $u_i(x)\le w_i(x)$ for $x\in J$ and $i\in\{1,2\}$.
\end{itemize}
\end{lemma}
\begin{proof}
See, for instance, the comparison theorem on page 112 of \cite{Walter1998}.  Note, however, that $\fv$ is quasimonotone decreasing and that we have stated the lemma for a right-boundary-value problem in $(ii)$.
\end{proof}

We use the following Lemma in the proof of Theorem \ref{thm:VF}.

\begin{lemma}\label{lem:phiy_lim_zero}
Let $\varphi$ be as in Proposition \ref{prop:FBP-sol}.$(i)$. For any $\bet>0$, $\lim_{y\to0^+}\frac{\varphi(y)}{y^\bet} = +\infty$.
\end{lemma}
\begin{proof}
The statement is trivial if $\lim_{y\to0^+} \varphi(y)>0$; therefore, suppose $\lim_{y\to0^+} \varphi(y)=0$, and define
\begin{linenomath*}\begin{align}
	F(y):=\frac{\varphi(y)}{y^\bet},
\end{align}\end{linenomath*}
for $0<y<\ya$. Our goal is to show that $\lim_{y\to0^+} F(y)=+\infty$. Assume, on the contrary, $\lim_{y\to0^+} F(y)\ne+\infty$. We compute
\begin{linenomath*}\begin{align}
	F'(y) = \frac{\varphi(y)}{y^{\bet+1}}\left(\frac{\rho}{\kap}\left(\frac{\kap}{\rho} - H\right) - \bet\right),
\end{align}\end{linenomath*}
for $0<y<\ya$.  Because $\lim_{y\to0^+}H(y)=\kap/\rho$ by Proposition \ref{prop:FBP-sol}.$(i)$, there exists an $\epsilon>0$ such that $F(y)$ is decreasing for $y\in(0,\epsilon)$.  Because $F$ is decreasing and positive on $(0,\epsilon)$, and because we assume $\lim_{y\to0^+} F(y)\ne+\infty$, we must have  $\lim_{y\to0^+} F(y)= M$ for some constant $M>0$.  From L'H\^{o}pital's rule, \eqref{eq:phi-H-FBP-sys}, and $\lim_{y\to0^+}H(y)=\kap/\rho$, we  deduce
\begin{linenomath*}\begin{align}
	M=\lim_{y\to 0^+} F(y) = \lim_{y\to 0^+}\frac{\varphi'(y)}{\bet y^{\bet-1}} 
	=\lim_{y\to 0^+} \frac{\rho}{\bet\kap}F(y)\left(\frac{\kap}{\rho}-H(y)\right)
	= 0,
\end{align}\end{linenomath*}
which contradicts $M>0$. Thus, we must have $\lim_{y\to0^+} F(y)=+\infty$.
\end{proof}

\end{document}